%% file: scalars.tex
\documentclass[submission,copyright,creativecommons]{eptcs}
\providecommand{\event}{QPL 2015}
\pdfoutput=1

\usepackage{amsmath} 
\usepackage{amsthm} 
\usepackage{amssymb} 
\usepackage{bold-extra} 
\usepackage{hyperref} 
\usepackage{stmaryrd} 

\newenvironment{mitem}
{\begin{itemize}
  \setlength{\itemsep}{1pt}
  \setlength{\parskip}{0pt}
  \setlength{\parsep}{0pt}}
{\end{itemize}}

\newtheorem{thm}{Theorem}
\newtheorem*{thm*}{Theorem}
\newtheorem{prop}[thm]{Proposition}
\newtheorem*{prop*}{Proposition}
\newtheorem{lem}[thm]{Lemma}
\newtheorem{cor}[thm]{Corollary}
\theoremstyle{definition}
\newtheorem{dfn}{Definition}
\newtheorem*{ex}{Example}

\newcommand{\ZX}{\textsc{zx}}
\newcommand{\NOT}{\textsc{not}}

\newcommand{\abs}[1]{\left| #1 \right|}

\newcommand{\ZZ}{\mathbb{Z}}

\newcommand{\ket}[1]{\left| #1 \right>} 
\newcommand{\bra}[1]{\left< #1 \right|} 

\renewcommand{\t}[1]{\ensuremath{^{\otimes #1}}}

\input{tikzit_header2.0.tex}

\title{Making the stabilizer \ZX-calculus complete for scalars}
\author{Miriam Backens
\institute{Department of Computer Science, University of Oxford, Parks Road, Oxford, OX1 3QD, UK}
\email{miriam.backens@cs.ox.ac.uk}
}
\def\titlerunning{Making the stabilizer \ZX-calculus complete for scalars}
\def\authorrunning{Miriam Backens}

\begin{document}

\maketitle

\begin{abstract}
 The \ZX-calculus is a graphical language for quantum processes with built-in rewrite rules.
 The rewrite rules allow equalities to be derived entirely graphically, leading to the question of completeness: can any equality that is derivable using matrices also be derived graphically?
 
 The \ZX-calculus is known to be complete for scalar-free pure qubit stabilizer quantum mechanics, meaning any equality between two pure stabilizer operators that is true up to a non-zero scalar factor can be derived using the graphical rewrite rules.
 Here, we replace those scalar-free rewrite rules with correctly scaled ones and show that, by adding one new diagram element and a new rewrite rule, the calculus can be made complete for pure qubit stabilizer quantum mechanics with scalars.
 This completeness property allows amplitudes and probabilities to be calculated entirely graphically.
 We also explicitly consider stabilizer zero diagrams, i.e.\ diagrams that represent a zero matrix, and show that two new rewrite rules suffice to make the calculus complete for those.
\end{abstract}

\section{Introduction}

Graphical languages provide high-level and intuitive formalisms for quantum processes.
The \ZX-calculus is one such language for pure qubit quantum mechanics, based on category theory \cite{coecke_interacting_2008}.
It has some features in common with quantum circuit notation, but unlike the former it also seamlessly integrates the representation of states and post-selected measurements.
Furthermore, the \ZX-calculus comes with a built-in set of rewrite rules that allow equalities to be derived graphically without the need to translate back to matrices.

The existence of rewrite rules for the \ZX-calculus means that this graphical calculus can potentially be used to replace matrix-based formalisms entirely for certain classes of problems.
However, this replacement is only possible without losing deductive power if the \ZX-calculus is \emph{complete}, i.e.\ if any equality that is derivable using matrices can also be derived graphically.

The overall \ZX-calculus for pure state qubit quantum mechanics is incomplete, and it is not obvious how to complete it \cite{schroeder_incomplete_2014}.
Yet the scalar-free \ZX-calculus for pure qubit stabilizer quantum mechanics is known to be complete \cite{backens_zx-calculus_2013}.
Stabilizer quantum mechanics is a fragment of quantum theory that exhibits many important quantum properties, like entanglement and non-locality.
It is of central importance in areas such as quantum error correcting codes \cite{nielsen_quantum_2010} and measurement-based quantum computation \cite{raussendorf_one-way_2001}.

There is one class of diagrams that the completeness proof in \cite{backens_zx-calculus_2013} ignores completely: those representing scalars.
In fact, some of the rewrite rules used in the proof are only true up to a non-zero scalar factor.
The existing completeness proof also passes over the issue of \emph{zero diagrams}, diagrams representing a zero matrix, relying implicitly on the fact that, within the stabilizer fragment of the \ZX-calculus, such diagrams can be easily recognised.

Here, we put the scalars back into the \ZX-calculus and show how to extend the completeness result to pure scalar diagrams as well as diagrams consisting of a scalar part and a non-scalar part.
To do this, we introduce a new type of \ZX-diagram element called the \emph{star node} and a new rewrite rule relating the star node to existing \ZX-calculus elements.
The new completeness result enables the stabilizer \ZX-calculus to be used to compute amplitudes and probabilities.
We also show explicitly how to recognise stabilizer zero diagrams.
The introduction of two new rewrite rules for zero diagrams allows the definition of a unique normal form for zero diagrams, and thus makes the \ZX-calculus complete for stabilizer zero diagrams.

The background for the new completeness proofs is given in Section \ref{s:ZX-calculus}.
This includes a quick introduction to the elements and rules of the \ZX-calculus, as well as an overview of the original stabilizer completeness proof.
Section \ref{s:scalars} contains the definitions of new \ZX-calculus components and rules.
The new completeness results can be found in section \ref{s:completeness}.
An example application for the new results is given in Section \ref{s:example}.
Finally, we state some conclusions and ideas for further work in Section \ref{s:conclusions}.

\section{The stabilizer \ZX-calculus}
\label{s:ZX-calculus}

The \ZX-calculus is a graphical language for pure state qubit quantum mechanics with post-selected measurements, introduced by Coecke and Duncan \cite{coecke_interacting_2008,coecke_interacting_2011} and expanded by Duncan and Perdrix \cite{duncan_graph_2009}.
Here, we are considering only the \ZX-calculus for stabilizer quantum mechanics \cite{nielsen_quantum_2010}, the fragment of pure qubit QM consisting of the following set of operations:
\begin{mitem}
 \item preparation of states in the computational basis
 \item unitary Clifford operations, generated by the single-qubit Hadamard operator, $H$, and the phase operator, $S$,
  \[
   H = \frac{1}{\sqrt{2}} \begin{pmatrix}1&1\\1&-1\end{pmatrix} \quad \text{and} \quad S = \begin{pmatrix}1&0\\0&i\end{pmatrix},
  \]
  as well as the two-qubit controlled-\NOT{} operator,
  \[
   C_X = \begin{pmatrix} 1&0&0&0 \\ 0&1&0&0 \\ 0&0&0&1 \\0&0&1&0 \end{pmatrix}
  \]
 \item measurements in the computational basis.
\end{mitem}
In the following sections, we give a quick introduction to the elements of a stabilizer \ZX-calculus diagram, their interpretations, and the rewrite rules.
More details about the \ZX-calculus can be found in \cite{coecke_interacting_2011,backens_zx-calculus_2013}.

\subsection{Elements of the stabilizer \ZX-calculus and the interpretation of diagrams}

A stabilizer \ZX-calculus diagram consists of different types of nodes connected by edges.
There are three types of nodes: green nodes with arbitrary numbers of inputs and outputs and a phase label $\alpha\in\{0,\pi/2,\pi,-\pi/2\}$, red nodes with arbitrary numbers of inputs and outputs and a phase label $\beta\in\{0,\pi/2,\pi,-\pi/2\}$, and yellow Hadamard nodes, which always have exactly one input and one output.

The matrix corresponding to a \ZX-calculus diagram $D$ is denoted by $\llbracket D \rrbracket$.
The basic nodes have the following interpretations:
\[
 \left\llbracket \input{tikz_files/green_spider.tikz} \right\rrbracket = \ket{0}\t{m}\bra{0}\t{n} + e^{i\alpha} \ket{1}\t{m}\bra{1}\t{n}, \qquad
 \left\llbracket \input{tikz_files/red_spider.tikz} \right\rrbracket = \ket{+}\t{l}\bra{+}\t{k} + e^{i\beta} \ket{-}\t{l}\bra{-}\t{k}, \text{ and}
\]
\[
 \left\llbracket \Hadamard \right\rrbracket = \ket{+}\bra{0} + \ket{-}\bra{1},
\]
where $\ket{\pm} = \frac{1}{\sqrt{2}}\left(\ket{0}\pm\ket{1}\right)$, and the zero-fold tensor product of any normalised bra or ket is taken to be the scalar 1.

Diagrams are read from bottom to top, i.e.\ edges connecting to a node or diagram at the bottom are inputs, edges connecting at the top are outputs.
Zero phase angles are often left out, i.e.\ \effect{rn} is the same as \effect{rn, label={[rphase]right:$0$}}.
A diagram with no inputs and at least one output is a \emph{state}, e.g. \state{gn, label={[gphase]right:$\pi$}} or \input{tikz_files/cup_diagram.tikz}.
Similarly, a diagram with neither inputs nor outputs is a \emph{scalar}. Examples of scalar diagrams are \scalar{rn, label={[rphase]right:$-\pi/2$}},
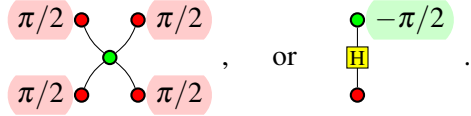
\begin{center}
 \input{tikz_files/scalar_ex.tikz}, $\quad$ or $\quad$ \input{tikz_files/scalar_ex2.tikz}.
\end{center}
A diagram with inputs but no outputs is an \emph{effect}.
Red or green nodes with exactly one input and one output are called \emph{phase shifts}, e.g. \phase{gn, label={[gphase]right:$-\pi/2$}}.

Putting diagrams side by side corresponds to taking the tensor product of their matrix equivalents; connecting inputs and outputs corresponds to matrix multiplication.
The Hermitian adjoint of a diagram can be found by flipping the original diagram upside-down and flipping the signs of all phase angles.

\subsection{Rewrite Rules}
\label{s:rewrite_rules}

The \ZX-calculus was originally defined with rewrite rules that are true up to a complex phase $e^{i\phi}$ for some $\phi\in[0,2\pi)$ \cite{coecke_interacting_2011}.
The normalising factors were later dropped, meaning that any equalities derived graphically hold up to some non-zero complex scalar factor.

The standard rewrite rules of the \ZX-calculus include two implicit rules:
\begin{center}
 ``ignore all non-zero scalars'' and ``only the topology matters''.
\end{center}
The former rule allows any non-zero scalar subdiagram to be replaced by the empty diagram; when we put the scalars back, in the later parts of this paper, we will obviously no longer be using this rule.
The topology rule means that the interpretation of a diagram depends only on the elements and the structure of their connections: nodes can be moved around and edges can be bent or straightened without changing the meaning of the diagram.
This rule remains valid in the scaled \ZX-calculus.

We now give the standard set of rewrite rules for the \ZX-calculus, using the following two short-hand conventions: firstly, any rule also holds with the colours red and green swapped.
Secondly, any rule also holds when flipped upside-down.

\textbf{Spider rule, loop rule, and cup rule}: Two nodes of the same colour can merge if they are connected by an edge, in that case their phases add.
Self-loops can be removed.
A node with two outputs and no inputs is the same as an edge bent into the shape of a cup.
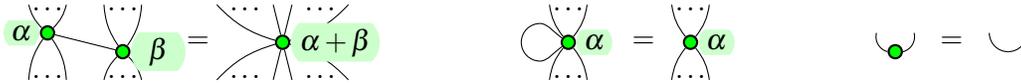
\begin{center}
 \input{tikz_files/spider2.0.tikz} $\qquad\qquad$ \input{tikz_files/loop2.0.tikz} $\qquad\qquad$ \input{tikz_files/cup2.0.tikz}
\end{center}

\textbf{Bialgebra rule, copy rule, and colour change rule}: The bialgebra rule allows a certain pattern of two red and two green nodes to be replaced by just one red and green node.
A node of one colour with one input and two outputs copies the zero phase state of the other colour.
The Hadamard node swaps the colour of red and green nodes when it is applied to each input and output.
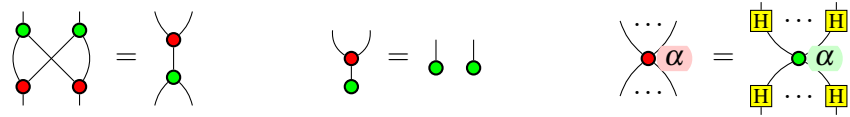
\begin{center}
 \input{tikz_files/bialgebra2.0.tikz} $\qquad\qquad$ \input{tikz_files/copy2.0.tikz} $\qquad\qquad$ \input{tikz_files/colour2.0.tikz}
\end{center}

\textbf{$\pi$-copy rule, $\pi$-commutation rule, and Euler decomposition of the Hadamard operator}: A $\pi$ phase shift is copied by a node of the other colour.
It can also be moved past any phase shift of the other colour, flipping the sign of that phase in the process.
The Hadamard node can be replaced by three green and red nodes of alternating colours, each with phase $\pi/2$.
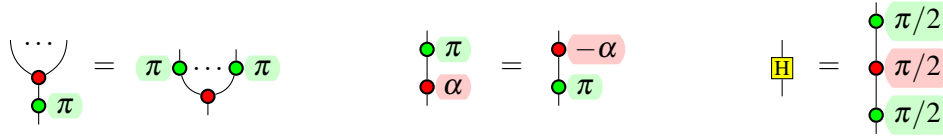
\begin{center}
 \input{tikz_files/pi-copy2.0.tikz} $\qquad\qquad$ \input{tikz_files/pi-comm2.0.tikz} $\qquad\qquad$ \input{tikz_files/Euler_dec2.0.tikz}
\end{center}

Whenever a rule holds for any number of edges, that number may be zero.
For example, the colour change rule with zero inputs and zero outputs implies that \scalar{rn, label={[rphase]right:$\alpha$}} = \scalar{gn, label={[gphase]right:$\alpha$}} for any $\alpha$.

Further rewrite rules can be derived from the above, e.g.\ the fact that \Hadamard{} is self-inverse.

\subsection{The original stabilizer \ZX-calculus completeness proof and two corollaries}

By completeness, we mean the following concept.

\begin{dfn}
 The \ZX-calculus is \emph{complete} for a subtheory of pure qubit quantum mechanics if any equality that holds in the matrix formulation of that theory can also be derived graphically.
\end{dfn}

The scalar-free stabilizer \ZX-calculus completeness proof \cite{backens_zx-calculus_2013} is based around a normal form for stabilizer state diagrams, i.e.\ diagrams that have outputs but no inputs.
This normal form is not unique, but it is straightforward to derive any equalities between diagrams given in normal form. 
By map-state duality, the completeness result extends to diagrams that have inputs as well as, or instead of, outputs.

The original completeness proof drops any scalar subdiagrams by the ``ignore all non-zero scalars'' meta-rule.
If we instead keep track of scalars, we can use the normal form proof to derive  results about scalar diagrams via the following trick: rewrite any connected scalar diagram into the inner product between some (possibly complicated) one-qubit state and \effect{gn}.
The ``state'' part can then be normalised using the algorithms given in the completeness proof.
This approach leads to the following results, which were implicit in \cite{backens_zx-calculus_2013} and are stated explicitly here.
Proofs of the corollaries are given in Appendix \ref{s:corollary_proofs}.

\begin{cor}\label{cor:decompose_scalars}
 Any stabilizer scalar diagram can be decomposed into disconnected segments containing at most two nodes each.
\end{cor}

\begin{cor}\label{cor:recognize_zero}
 When a stabilizer zero diagram, i.e.\ any stabilizer diagram that represents a zero matrix, is brought into normal form and all scalar subdiagrams are decomposed as in Corollary \ref{cor:decompose_scalars}, the resulting diagram explicitly contains at least one of
 \begin{equation}\label{eq:zero_scalars}
  \scalar{gn, label={[gphase]right:$\pi$}}, \quad\quad \scalar{rn, label={[rphase]right:$\pi$}}, \quad\quad \innerprod{gn, label={[gphase]right:$\pi/2$}}{rn, label={[rphase]right:$-\pi/2$}}, \quad \text{ or } \quad \innerprod{gn, label={[gphase]right:$-\pi/2$}}{rn, label={[rphase]right:$\pi/2$}}.
 \end{equation}
\end{cor}

In the scalar-free \ZX-calculus, it is straightforward to show that the four scalars given in \eqref{eq:zero_scalars} are all equal to each other.
We will show in Section \ref{s:zero_completeness} that they can also be rewritten into each other in the scaled \ZX-calculus.

Corollary \ref{cor:decompose_scalars} implies that in order to prove completeness for stabilizer scalar diagrams, it suffices to look only at diagrams made up of small disconnected segments.
Corollary \ref{cor:recognize_zero} shows that there is an algorithm for deciding whether a diagram in the stabilizer \ZX-calculus is a zero diagram.
This is necessary if we want to consider non-zero diagrams only, and has already been implicitly assumed in the scalar-free completeness proof.

\section{Putting the scalars back in}
\label{s:scalars}

Building on the scalar-free stabilizer \ZX-calculus completeness proof, we now show how to ``put the scalars back in'' and make the stabilizer \ZX-calculus complete for scalar diagrams and scaled operators, as well as for zero diagrams.
To do this, we first modify the rewrite rules to be exactly true (not just up to scalar factor).
We then argue that a new node is required to make the calculus complete for non-zero scalars.
Introduction of this node and a new rewrite rule for it allows us to prove completeness for non-zero stabilizer scalars.
Finally, we prove a normal form theorem for zero diagrams, based on a further two new rewrite rules specifically for diagrams containing the zero scalar.

\subsection{Normalisation of scalar diagrams in the \ZX-calculus}
\label{s:normalisation}

The basic states and effects in the \ZX-calculus are not normalised to 1, e.g.:
\begin{equation}
 \left\llbracket \state{gn} \right\rrbracket = \sqrt{2}\ket{+} \quad\text{and}\quad
 \left\llbracket \effect{gn} \right\rrbracket = \sqrt{2}\bra{+},
\end{equation}
and similarly for nodes with more outputs.
This normalisation was chosen to avoid having to include normalising factors in the spider law.
As a result, inner products of a red state and green effect (or conversely) are either zero or have norm strictly greater than 1.
For example:
\begin{equation}
 \left\llbracket\innerprodgr\right\rrbracket = \sqrt{2}, \qquad
 \left\llbracket\innerprod{gn,label={[gphase]right:$\alpha$}}{rn,label={[rphase]right:$\pi$}}\right\rrbracket = \sqrt{2}e^{i\alpha}, \quad \text{and} \quad
 \left\llbracket\innerprod{gn,label={[gphase]right:$-\pi/2$}}{rn,label={[rphase]right:$-\pi/2$}}\right\rrbracket = 2e^{-i\pi/4}. \label{eq:scalars}
\end{equation}
The node-free loop \input{tikz_files/loop_scalar2.0.tikz} is also a scalar.
By the cup rule and the loop rule:
\begin{equation}
 \left\llbracket\input{tikz_files/loop_scalar2.0.tikz}\right\rrbracket = \left\llbracket\input{tikz_files/loop_gn.tikz}\right\rrbracket = \left\llbracket\scalar{gn}\right\rrbracket = 2.
\end{equation}

\subsection{Modified rewrite rules}

Some of the rewrite rules for the \ZX-calculus as given in section \ref{s:rewrite_rules} are sound only when equality is taken up to a non-zero scalar factor.
Other rewrite rules are already correctly scaled.
We now give the corrected versions of those rewrite rules that need modification.
For some cases, this means going back to the original form of the rewrite rules introduced  in \cite{coecke_interacting_2011}, before scalars were dropped.
This applies to the bialgebra rule and the copy rule:
\begin{center}
 \input{tikz_files/bialgebra_scalars2.0.tikz} $\quad$ and $\quad$ \input{tikz_files/copy_scalars2.0.tikz},
\end{center}
as well as the \emph{Hopf law}\footnote{The Hopf law can be derived from the other rules of the \ZX-calculus, which is why it is not included in Section \ref{s:rewrite_rules}.}:
\begin{equation}\label{rr:Hopf}
 \input{tikz_files/Hopf_scalars2.0.tikz}\;.
\end{equation}

Two other rules are normalised, but require the introduction of a complex phase factor in order to hold exactly.
These rules were originally introduced without any scalar factors.
There is no unique way of making them exact: the complex phase could be put on one side of the equation, or the other, or even split between the two.
We choose the convention of putting the complex phase on the right-hand side, resulting in the following rules:
\begin{mitem}
 \item the $\pi$-commutation law:
  \begin{equation}\label{rr:pi-comm}
   \input{tikz_files/pi-comm_scalars2.0.tikz},
  \end{equation}
 \item and the Euler decomposition of the Hadamard:
  \begin{equation}\label{rr:Euler}
   \input{tikz_files/Euler_dec_scalars2.0.tikz}.
  \end{equation}
\end{mitem}
Additional real scalars have been added to the left-hand side of the equations to preserve the normalisation, cf. \eqref{eq:scalars}.

The other rules listed in section \ref{s:rewrite_rules} are correctly scaled in the form given there.

Both corollaries of the original stabilizer completeness proof continue to hold under the scaled rewrite rules, unless keeping track of scalars stops us from applying a rewrite rule whose unscaled counterpart does apply.
\begin{ex}
 Consider the following diagram:
 \begin{equation}\label{eq:copy_example}
  \input{tikz_files/copy_example.tikz}\;.
 \end{equation}
 Using the scalar-free rewrite rules, this can be rewritten to \state{gn} \state{gn}. Yet using the scaled rewrite rules, the diagram cannot be rewritten, as the copy rule now requires a factor of \innerprodgr.
\end{ex}
If all scalars are invertible, this problem disappears: given an inverse for \innerprodgr, we could just add both that scalar and the inverse to the diagram in \eqref{eq:copy_example} and then apply the copy rule, as in the following example.

\begin{ex}
 Suppose the scalar diagram \genscalar{$s$} is the inverse of \innerprodgr. Graphically, this is expressed as follows:
 \begin{equation}\label{eq:innerprodgr_inverse}
  \genscalar{$s$} \; \innerprodgr \; = \quad ,
 \end{equation}
 since the empty diagram represents the scalar identity: $\llbracket\;\;\rrbracket = 1$.
 Using \eqref{eq:innerprodgr_inverse} and the fact that any \ZX-calculus diagram contains the empty diagram as a subdiagram, we can derive:
 \begin{equation}
  \input{tikz_files/copy_example.tikz} \;
  = \; \genscalar{$s$} \; \innerprodgr \; \input{tikz_files/copy_example.tikz} \;
  = \; \genscalar{$s$} \; \state{gn} \; \state{gn}.
 \end{equation}
 The non-scalar part of the diagram has thus been transformed in the desired way.
\end{ex}

As seen in section \ref{s:normalisation}, no simple \ZX-calculus diagram built from just a few of the existing nodes can be the inverse to \innerprodgr.
We will therefore introduce a new node type.
To keep the number of rewrite rules small, it will be more convenient to define this new node as the inverse of \scalar{gn} rather than \innerprodgr; the inverse for the latter can then be derived.
\begin{dfn}\label{dfn:halfscalar}
 Let \halfscalar{} be a new node type called \emph{star node} with no inputs or outputs. The star node is defined as the inverse of \scalar{gn}, i.e.\ it obeys the following rewrite rule, called \emph{star rule}:
 \begin{equation}
   \halfscalar \; \scalar{gn} \; = \quad .
 \end{equation}
\end{dfn}

\begin{lem}\label{lem:halfscalar_interpretation}
 Definition \ref{dfn:halfscalar} is consistent with the interpretation $\left\llbracket \halfscalar \right\rrbracket = \frac{1}{2}$.
\end{lem}

\begin{lem}
 The star node \halfscalar{} is also inverse to \innerprodgr{} \innerprodgr{} :
 \begin{equation}\label{eq:halfscalar_innerprodgr2}
  \halfscalar \; \innerprodgr \; \innerprodgr \; = \quad .
 \end{equation}
\end{lem}
\begin{proof}
 Using the star rule, loop rule, cup rule, spider rule, and Hopf rule, we have:
 \begin{equation}
  \innerprodgr \; \innerprodgr \;
  = \; \halfscalar \; \innerprodgr \; \innerprodgr \; \scalar{gn} \;
  = \; \halfscalar \; \innerprodgr \; \innerprodgr \; \input{tikz_files/inverse_der4.tikz} \;
  = \; \halfscalar \; \innerprodgr \; \innerprodgr \; \input{tikz_files/inverse_der3.tikz} \;
  =  \; \halfscalar \; \innerprodgr \; \innerprodgr \; \input{tikz_files/inverse_der2.tikz} \;
  =  \; \halfscalar \; \input{tikz_files/inverse_der1.tikz} \;
  = \; \halfscalar \; \scalar{gn} \; \scalar{rn} \;
  = \; \scalar{gn}.
 \end{equation}
 Thus the desired equality follows directly from the star rule.
\end{proof}

With \halfscalar{} and \eqref{eq:halfscalar_innerprodgr2}, application of the scaled bialgebra, copy, and Hopf rules can never fail because of normalisation problems.
In the following section, we show that, given the star node and star rule, the scalars appearing in the $\pi$-commutation and Euler rules also have inverses.
Thus the completeness proof for non-scalar diagrams and its corollaries still hold if we replace the original rewrite rules with the scaled versions and add \halfscalar{} and the star rule to the calculus.

\section{The completeness results}
\label{s:completeness}

\subsection{A normal form for non-zero stabilizer scalars}

In this section, whenever we talk about scalars, we mean non-zero stabilizer scalars.
Using Corollary \ref{cor:decompose_scalars} and Lemma \ref{lem:halfscalar_interpretation}, it is straightforward to show the following.
\begin{prop}
 Let $D$ be a non-zero stabilizer scalar diagram.
 Then:
 \begin{equation}\label{eq:scalar_values}
  \left\llbracket D \right\rrbracket \in \left\{ \sqrt{2^r}e^{is\pi/4} \right|\left.\vphantom{e^{is\pi/4}\sqrt{2^r}} r,s\in\ZZ\right\}.
 \end{equation}
\end{prop}

We will define a normal form for scalar diagrams as follows: Pick one representative diagram for each $s\in\{0,1,\ldots,7\}$ from \eqref{eq:scalar_values}.
Combine these with the smallest number of copies of \halfscalar{} and/or \innerprodgr{} required to achieve the correct normalisation.

The simplest representative for $s=0$ is the empty diagram.

\begin{lem}
 The set
 \begin{equation}\label{eq:complex_phases}
  \left\{ \innerprod{gn,label={[gphase]right:$\pi/2$}}{rn,label={[rphase]right:$\pi/2$}}, \quad \innerprod{gn,label={[gphase]right:$\pi/2$}}{rn,label={[rphase]right:$\pi$}}, \quad \innerprod{gn,label={[gphase]right:$\pi/2$}}{rn,label={[rphase]right:$\pi$}} \innerprod{gn,label={[gphase]right:$\pi/2$}}{rn,label={[rphase]right:$\pi/2$}}, \quad \innerprod{gn,label={[gphase]right:$\pi$}}{rn,label={[rphase]right:$\pi$}}, \quad \innerprod{gn,label={[gphase]right:$-\pi/2$}}{rn,label={[rphase]right:$\pi$}} \innerprod{gn,label={[gphase]right:$-\pi/2$}}{rn,label={[rphase]right:$-\pi/2$}}, \quad \innerprod{gn,label={[gphase]right:$-\pi/2$}}{rn,label={[rphase]right:$\pi$}}, \quad \innerprod{gn,label={[gphase]right:$-\pi/2$}}{rn,label={[rphase]right:$-\pi/2$}} \right\}
 \end{equation}
 contains one diagram for each complex phase $e^{is\pi/4}$ with $s\in\{1,\ldots,7\}$.
\end{lem}
\begin{proof}
 Apply the interpretation map to each diagram in turn, ignoring normalisation.
\end{proof}

\begin{lem}\label{lem:modulus_nf}
 Let $r$ be an integer.
 Then the scalar $\sqrt{2^r}$ can be represented by one of the following \ZX-calculus diagrams: 
 \begin{mitem}
  \item for $r>0$, $r$ copies of \innerprodgr,
  \item for $r=0$, the empty diagram,
  \item for $r<0$ and $r$ even, $\abs{r}/2$ copies of \halfscalar, and
  \item for $r<0$ and $r$ odd, one copy of \innerprodgr{} and $(1-r)/2$ copies of \halfscalar.
 \end{mitem}
\end{lem}
\begin{proof}
 It is straightforward to check the interpretations of the given diagrams.
\end{proof}

In fact, any diagram consisting solely of copies of \innerprodgr{} and \halfscalar{} can be brought into the normal form given in Lemma \ref{lem:modulus_nf} by applying \eqref{eq:halfscalar_innerprodgr2} as many times as possible.

We will now state and prove the normal form theorem for non-zero scalars.
For ease of understanding, some parts of the proof are given as separate lemmas, which are stated after the main theorem.
The proofs of those lemmas, as well as a full proof of the main theorem, can be found in Appendix \ref{s:appendix}.

\begin{thm}\label{thm:scalar_nf}
 Any non-zero stabilizer scalar diagram in the \ZX-calculus can be uniquely represented as a combination of one of the diagrams in \eqref{eq:complex_phases} or the empty diagram with one of the diagrams listed in Lemma \ref{lem:modulus_nf}.
\end{thm}
\begin{proof}[Proof (sketch, full proof in Appendix \ref{s:appendix})]
 By Corollary \ref{cor:decompose_scalars} and Lemma \ref{lem:y-states}, any scalar diagram can be written in terms of star nodes and diagrams of the form
 \begin{center}
  \innerprod{gn,label={[gphase]right:$\alpha$}}{rn,label={[rphase]right:$\beta$}}
 \end{center}
 for some $\alpha,\beta\in\{0,\pi/2,\pi,-\pi/2\}$.
 Using Lemmas \ref{lem:innerprod_wlog} and \ref{lem:overlap_with_ket_zero}, any scalar diagram can be further rewritten to a combination of elements of the following set:
 \begin{equation}
  \left\{ \halfscalar, \quad \innerprod{gn}{rn}, \quad \innerprod{gn,label={[gphase]right:$\pi/2$}}{rn,label={[rphase]right:$\pi/2$}}, \quad
  \innerprod{gn,label={[gphase]right:$\pi/2$}}{rn,label={[rphase]right:$\pi$}}, \quad
  \innerprod{gn,label={[gphase]right:$\pi$}}{rn,label={[rphase]right:$\pi$}}, \quad
  \innerprod{gn,label={[gphase]right:$-\pi/2$}}{rn,label={[rphase]right:$\pi$}}, \quad
  \innerprod{gn,label={[gphase]right:$-\pi/2$}}{rn,label={[rphase]right:$-\pi/2$}} \right\}.
 \end{equation}
 Any diagram built from these components can be brought into the desired normal form via Lemmas \ref{lem:pi_multiplication} and \ref{lem:omega_inverses}.
\end{proof}

\begin{lem}\label{lem:innerprod_wlog}
 The inner product between a red and a green node with phase angles $\alpha$ and $\beta$ is defined only by the set $\{\alpha,\beta\}$, it does not matter which label is assigned to the green and which to the red node, or whether it is a green state and red effect, or conversely. Diagrammatically:
 \begin{equation}
  \innerprod{gn, label={[gphase]right:$\alpha$}}{rn, label={[rphase]right:$\beta$}} \; = \; \innerprod{rn, label={[rphase]right:$\beta$}}{gn, label={[gphase]right:$\alpha$}} \; = \; \innerprod{gn, label={[gphase]right:$\beta$}}{rn, label={[rphase]right:$\alpha$}} \; = \; \innerprod{rn, label={[rphase]right:$\alpha$}}{gn, label={[gphase]right:$\beta$}}.
 \end{equation}
\end{lem}

\begin{lem}\label{lem:pi_multiplication}
 For any pair of phase angles $\alpha$ and $\beta$, the complex phases resulting from the inner product of \state{rn, label={[rphase]right:$\pi$}} with phased green effects can be combined into just one subdiagram and a normalising factor \innerprodgr:
 \begin{equation}
  \innerprod{gn, label={[gphase]right:$\alpha$}}{rn, label={[rphase]right:$\pi$}} \; \innerprod{gn, label={[gphase]right:$\beta$}}{rn, label={[rphase]right:$\pi$}} \; = \; \innerprod{gn}{rn} \; \innerprod{gn, label={[gphase]right:$\alpha+\beta$}}{rn, label={[rphase]right:$\pi$}}.
 \end{equation}
\end{lem}

\begin{lem}\label{lem:overlap_with_ket_zero}
 For any phase angle $\alpha$, the inner product of a green effect of phase $\alpha$ with \state{rn} is equal to \innerprodgr{}:
 \begin{equation}
  \innerprod{gn, label={[gphase]right:$\alpha$}}{rn} \; = \; \innerprod{gn}{rn}. 
 \end{equation}
\end{lem}

\begin{lem}\label{lem:y-states}
 The states \state{gn, label={[gphase]right:$\pi/2$}} and \state{rn, label={[rphase]right:$-\pi/2$}} are equal up to a complex phase:
 \begin{equation}
  \state{rn, label={[rphase]right:$-\pi/2$}} = \halfscalar \; \innerprod{gn, label={[gphase]right:$-\pi/2$}}{rn, label={[rphase]right:$-\pi/2$}} \; \state{gn, label={[gphase]right:$\pi/2$}}.
 \end{equation}
\end{lem}

\begin{lem}\label{lem:omega_inverses}
 The scalar diagrams \halfscalar{} \innerprod{gn,label={[gphase]right:$-\pi/2$}}{rn,label={[rphase]right:$-\pi/2$}} and \halfscalar{} \innerprod{gn,label={[gphase]right:$\pi/2$}}{rn,label={[rphase]right:$\pi/2$}} are inverse to each other:
 \begin{equation}
  \halfscalar \innerprod{gn,label={[gphase]right:$-\pi/2$}}{rn,label={[rphase]right:$-\pi/2$}} \halfscalar \innerprod{gn,label={[gphase]right:$\pi/2$}}{rn,label={[rphase]right:$\pi/2$}} = \quad.
 \end{equation}
\end{lem}

\subsection{Completeness for scaled operators}

Given the completeness result for non-scalar stabilizer \ZX-calculus diagrams \cite{backens_zx-calculus_2013} and the normal form for scalar diagrams derived in Theorem \ref{thm:scalar_nf}, we can now derive equalities between diagrams that have both scalar and non-scalar parts.

\begin{thm}
 The \ZX-calculus with scaled rewrite rules, including \halfscalar{} and the star rule, is complete for non-zero scaled diagrams, i.e.\ diagrams that contain both scalar and non-scalar parts.
\end{thm}
\begin{proof}
 Given two scaled diagrams $D$ and $D'$, first apply the algorithm from the completeness proof for non-scalar diagrams, except rather than dropping the scalars, keep track of them.
 Use the invertibility of scalars to make sure that any rewrite rule required for this process can be applied.
 If the non-scalar parts of the diagrams are not equal, clearly the scaled diagrams are not equal.
 Otherwise, proceed by bringing the scalars into normal form as described in Theorem \ref{thm:scalar_nf}.
 This normalisation does not change the non-scalar part of the diagram.

 If the two resulting diagrams are equal, some of the rewrite steps can be inverted to find a series of rewrites transforming $D$ into $D'$ (or conversely), thus proving that the two diagrams are equal according to the rules of the graphical calculus.
 Otherwise, the two diagrams must represent different operators, as multiplying the same operator by two different scalars yields two different operators.
\end{proof}

\subsection{Completeness for stabilizer zero diagrams}
\label{s:zero_completeness}

We have shown that the stabilizer \ZX-calculus is complete for scaled diagrams as long as they are non-zero.
By Corollary \ref{cor:recognize_zero}, any stabilizer zero diagram can be rewritten to explicitly contain one of the following scalar diagrams as a subdiagram:
\begin{equation}\label{eq:zero_scalars2}
 \scalar{gn, label={[gphase]right:$\pi$}}, \quad\quad \scalar{rn, label={[rphase]right:$\pi$}}, \quad\quad \innerprod{gn, label={[gphase]right:$\pi/2$}}{rn, label={[rphase]right:$-\pi/2$}}, \quad \text{ or } \quad \innerprod{gn, label={[gphase]right:$-\pi/2$}}{rn, label={[rphase]right:$\pi/2$}}.
\end{equation}
Of course the calculus does not actually contain four distinct representations of 0, as shown in the following lemma.

\begin{lem}\label{lem:unique_zero}
 Any diagram that contains one of the subdiagrams in \eqref{eq:zero_scalars2} can be rewritten to contain \scalar{gn, label={[gphase]right:$\pi$}}.
\end{lem}
\begin{proof}
 By the colour change law, $\scalar{rn, label={[rphase]right:$\pi$}} \; = \; \scalar{gn, label={[gphase]right:$\pi$}}$.
 Using Lemma \ref{lem:y-states}, we find that
 \begin{equation}
  \innerprod{gn, label={[gphase]right:$\pi/2$}}{rn, label={[rphase]right:$-\pi/2$}} \; = \; \halfscalar \; \innerprod{gn, label={[gphase]right:$-\pi/2$}}{rn, label={[rphase]right:$-\pi/2$}} \; \scalar{gn, label={[gphase]right:$\pi$}}.
 \end{equation}
 This result also applies to \innerprod{gn, label={[gphase]right:$-\pi/2$}}{rn, label={[rphase]right:$\pi/2$}} by Lemma \ref{lem:innerprod_wlog}.
\end{proof}

Given the set of rewrite rules listed so far, the stabilizer \ZX-calculus cannot be complete for zero diagrams as there are no specific rewrite rules involving the zero scalar.
The following new rewrite rule was suggested to resolve this issue \cite{kissinger_communication_2014}:
\begin{equation}
 \input{tikz_files/zero2.0.tikz}.
\end{equation}
This \emph{zero rule} allows any zero diagram to be transformed into a completely disconnected diagram.
In the scalar-free \ZX-calculus, which was the original context for the introduction of the zero rule, that is sufficient to derive a normal form for zero diagrams.
To make the \emph{scaled} stabilizer \ZX-calculus complete for zero diagrams, we will need to introduce an additional \emph{zero scalar rule}:
\begin{equation}
 \scalar{gn, label={[gphase]right:$\pi$}} \; \scalar{gn, label={[gphase]right:$\alpha$}} \; = \; \scalar{gn, label={[gphase]right:$\pi$}}
\end{equation}
for any $\alpha$.

\begin{thm}\label{thm:zero_completeness}
 The scaled stabilizer \ZX-calculus with \halfscalar{}, the star rule, the zero rule, and the zero scalar rule is complete for zero diagrams.
\end{thm}
\begin{proof}
 We will show that with the given rules any zero diagram with $n$ inputs and $m$ outputs can be rewritten into the following normal form:
 \begin{equation}\label{eq:zero_normal_form}
  \input{tikz_files/zero_normal_form.tikz}
 \end{equation}
 First, note that the normal form is clearly unique as a zero matrix, the interpretation of any zero diagram, is fully determined by its dimensions.
 
 Now, to rewrite a zero diagram into normal form, first apply the Euler decomposition rule to remove all Hadamard nodes.
 Apply the spider rule until all remaining edges are between one red and one green node.
 Then apply the zero rule (or its upside-down equivalent, depending on how the colours match up) to each edge, transforming the diagram into a completely disconnected graph where the only remaining edges are inputs and outputs of the diagram.
 Further applications of the zero rule (upside-down or not) can be used to change the colour of the nodes connected to inputs or outputs.
 Other than one copy of \scalar{gn, label={[gphase]right:$\pi$}}, any disconnected red and green nodes can be removed using the zero scalar rule.
 Finally, any copies of \halfscalar{} can be removed by using the zero scalar rule to create a new copy of \scalar{gn}, and then applying the star rule.
 This leaves the diagram in the normal form \eqref{eq:zero_normal_form}.
 
 As all rewrite rules are invertible, this implies that any two zero diagrams with the same numbers of inputs and outputs can be rewritten into each other.
 Therefore the \ZX-calculus is complete for stabilizer zero diagrams.
\end{proof}

Many of the results derived in this paper are not actually specific to stabilizer diagrams.
In fact, given the zero and zero scalar rules, any diagram that explicitly contains \scalar{gn, label={[gphase]right:$\pi$}} can be brought into the normal form given in \eqref{eq:zero_normal_form}.
Nevertheless, Theorem \ref{thm:zero_completeness} only holds for stabilizer zero diagrams, as, e.g.\ using the incompleteness example in \cite{schroeder_incomplete_2014}, we can construct non-stabilizer zero diagrams which cannot be rewritten to contain \scalar{gn, label={[gphase]right:$\pi$}} as a subdiagram, and without a copy of \scalar{gn, label={[gphase]right:$\pi$}} the normalisation process described in the proof of Theorem \ref{thm:zero_completeness} does not work.

\section{Example: Quantum key distribution}
\label{s:example}

Consider the BB84 protocol for quantum key distribution using a two-qubit Bell state \cite{bennett_quantum_1984}: Alice and Bob each hold one half of the entangled state, and they each randomly decide to measure their qubit in either the computational basis $\{\ket{0}, \ket{1}\}$ or the Hadamard basis $\{\ket{+}, \ket{-}\}$.
They later compare their choice of basis over a classical communication channel, e.g.\ a telephone line.
Assuming the Bell state is $\frac{1}{\sqrt{2}}\left(\ket{00}+\ket{11}\right)$, if they have picked the same basis, their results always agree and they have a key; if they have picked different bases, their results are uncorrelated.
We are not interested in the security properties of this key distribution scheme here, instead we simply use it as motivation for some \ZX-calculus derivations since all the states and operations required for the protocol are in the stabilizer formalism.

The Bell state given above is represented in the \ZX-calculus as a ``cup'' with some normalising factors:
\begin{equation}
 \frac{1}{\sqrt{2}}\left(\ket{00}+\ket{11}\right) = \left\llbracket \halfscalar \; \innerprodgr \; \input{tikz_files/cup_diagram.tikz} \right\rrbracket.
\end{equation}
Measurements in the \ZX-calculus are post-selected, so we have to consider each pair of measurement outcomes in turn.
Graphically, the normalised outcomes of computational and Hadamard basis measurements on single qubits are:
\begin{equation}
 \bra{0}= \left\llbracket \halfscalar{} \; \innerprodgr{} \; \effect{rn} \right\rrbracket, \quad
 \bra{1} = \left\llbracket \halfscalar{} \; \innerprodgr{} \; \effect{rn, label={[rphase]right:$\pi$}} \right\rrbracket, \quad
 \bra{+}= \left\llbracket \halfscalar{} \; \innerprodgr{} \; \effect{gn} \right\rrbracket, \quad\text{and}\quad
 \bra{-} = \left\llbracket \halfscalar{} \; \innerprodgr{} \; \effect{gn, label={[gphase]right:$\pi$}} \right\rrbracket.
\end{equation}

First, assume Alice and Bob both measure in the computational basis. Can they both get outcome 0?
Constructing the \ZX-calculus diagram for the overlap between the Bell state and $\bra{00}$ and then simplifying it using \eqref{eq:innerprodgr_inverse}, the topology rule, the spider rule, and the star rule, yields:
\begin{equation}
 \halfscalar \; \halfscalar \; \halfscalar \; \innerprodgr \; \innerprodgr \; \innerprodgr \; \input{tikz_files/bell1.tikz} \;
 = \; \halfscalar \; \halfscalar \; \innerprodgr \; \innerprod{rn}{rn} \;
 = \; \halfscalar \; \halfscalar \; \innerprodgr \; \scalar{rn} \;
 = \; \halfscalar \; \innerprodgr,
\end{equation}
which is non-zero, so this outcome is possible.
The probability of this outcome can be found by multiplying this amplitude with its dagger, graphically:
\begin{equation}
 \halfscalar \; \input{tikz_files/innerprodrg.tikz} \; \halfscalar \; \innerprodgr \;
 = \; \halfscalar \; \halfscalar \; \innerprodgr \; \innerprodgr \;
 = \; \halfscalar,
\end{equation}
i.e.\ the probability of Alice and Bob both measuring 0 is $\llbracket\halfscalar\rrbracket=1/2$.
Similarly, the overlap of the Bell state with the effect $\bra{11}$ is, graphically:
\begin{equation}
 \halfscalar \; \halfscalar \; \halfscalar \; \innerprodgr \; \innerprodgr \; \innerprodgr \; \input{tikz_files/bell5.tikz} \;
 = \; \halfscalar \; \halfscalar \; \innerprodgr \; \innerprod{rn, label={[rphase]right:$\pi$}}{rn, label={[rphase]right:$\pi$}} \;
 = \; \halfscalar \; \halfscalar \; \innerprodgr \; \scalar{rn} \;
 = \; \halfscalar \; \innerprodgr,
\end{equation}
so the probability of Alice and Bob both getting measurement outcome 1 is again $1/2$.

On the other hand, if we consider whether Alice can get 0 while Bob gets 1, we find the following:
\begin{equation}
 \halfscalar \; \halfscalar \; \halfscalar \; \innerprodgr \; \innerprodgr \; \innerprodgr \; \input{tikz_files/bell2.tikz} \;
 = \; \halfscalar \; \halfscalar \; \; \innerprodgr \; \innerprod{rn}{rn,label={[rphase]right:$\pi$}} \;
 = \; \halfscalar \; \halfscalar \; \innerprodgr \; \scalar{rn,label={[rphase]right:$\pi$}} \;
 = \; \halfscalar \; \halfscalar \; \scalar{gn} \; \scalar{rn} \; \scalar{rn,label={[rphase]right:$\pi$}} \;
 = \; \scalar{gn,label={[gphase]right:$\pi$}},
\end{equation}
where the penultimate equality is by the zero rule and the last one by the star and colour change rules. As $\llbracket\scalar{gn,label={[gphase]right:$\pi$}}\rrbracket=0$, that combination of outcomes is impossible.

If Alice measures in the computational basis and Bob in the Hadamard basis, the probability of Alice getting outcome $\theta$ and Bob getting $\phi$ for some fixed $\theta,\phi\in\{0,\pi\}$ is:
\begin{equation}
 \begin{array}{c} \halfscalar \; \halfscalar \; \halfscalar \; \input{tikz_files/innerprodrg.tikz} \; \input{tikz_files/innerprodrg.tikz} \; \input{tikz_files/innerprodrg.tikz} \\ \halfscalar \; \halfscalar \; \halfscalar \; \innerprodgr \; \innerprodgr \; \innerprodgr \end{array} \input{tikz_files/bell3.tikz} \;
 = \; \halfscalar \; \halfscalar \; \halfscalar \; \innerprod{gn,label={[gphase]right:$-\phi$}}{rn,label={[rphase]right:$-\theta$}} \; \innerprod{gn,label={[gphase]right:$\phi$}}{rn,label={[rphase]right:$\theta$}} \;
 = \; \halfscalar \; \halfscalar \; \halfscalar \; \innerprod{gn}{rn} \; \innerprod{gn}{rn} \;
 = \; \halfscalar \; \halfscalar,
\end{equation}
where the penultimate equality is by Lemmas \ref{lem:pi_multiplication} and \ref{lem:overlap_with_ket_zero}.
Thus if Alice and Bob choose different bases, their outcomes are random and uncorrelated: each of the four combinations of outcomes has probability $\llbracket\halfscalar\;\halfscalar\rrbracket = 1/4$.

\section{Conclusions}
\label{s:conclusions}

Completeness is the property of a graphical language for quantum processes of having the same deductive power as matrix mechanics.
We show how the introduction of a new type of \ZX-calculus element, called star node, along with a new rewrite rule relating the star node to existing \ZX-calculus elements, allows the scalar-free stabilizer completeness proof to be extended to include non-zero scalars.
The \ZX-calculus can thus be used to compute amplitudes and probabilities.
Furthermore, we show that, as implicitly assumed in the scalar-free completeness proof, it is straightforward to recognise stabilizer zero diagrams, and that two new rewrite rules suffice to derive a unique normal form for these zero diagrams.

These results allow any problem in pure stabilizer quantum mechanics to be analysed entirely graphically, including questions about inner products between pure states or probabilities associated with pure projective measurements.
Thus, as a next step, procedures for solving such problems could be implemented in the automated graph rewriting system \texttt{Quantomatic} \cite{quantomatic}, which can rewrite \ZX-calculus diagrams either fully independently or in a user-guided manner.
It would also be interesting to see whether the completeness results can be extended to mixed states and completely positive operators, for example via the \emph{CPM construction} \cite{selinger_dagger_2007}.
Work is ongoing to prove completeness results for other fragments of pure state qubit quantum mechanics.

\section*{Acknowledgments}

I would like to thank Aleks Kissinger for pointing out the issue of zero diagram completeness and for constructive discussion on how to resolve it, and Dominic Horsman for helpful comments on this paper.
This work was supported by EPSRC and by the John Templeton Foundation.

\bibliographystyle{eptcs}
\bibliography{refs}

\appendix

\section{Proofs of corollaries to the scalar-free completeness result}
\label{s:corollary_proofs}

The corollaries rely not on the actual completeness proof in \cite{backens_zx-calculus_2013}, but on a normal form result used in that proof.
We recap the definition of the normal form, called GS-LC form, and the process for rewriting into normal form before proceeding to the proofs of the corollaries.

A diagram in GS-LC form consists of a graph state with single-qubit Clifford unitaries applied to the outputs (cf.\ Definition 11 in \cite{backens_zx-calculus_2013}).
A graph state diagram on $n$ qubits consists of $n$ green nodes with one output each (the nodes of the graph), and some number of Hadamard nodes, each connected to two distinct green nodes (the edges of the graph).
All GS-LC diagrams are states, but by the Choi-Jamio{\l}kowski isomorphism, this normal form and the associated results extend to arbitrary diagrams.
An example GS-LC diagram is shown in Figure \ref{fig:GS-LC}.

\begin{figure}
 \centering
 \input{tikz_files/example_GS-LC.tikz}
 \caption{A scalar-free stabilizer \ZX-calculus diagram in GS-LC form.}
 \label{fig:GS-LC}
\end{figure}
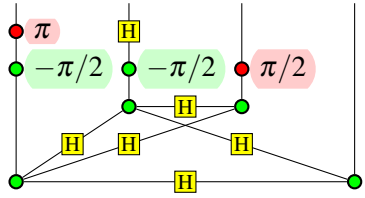

Any diagram in the scalar-free stabilizer \ZX-calculus can be brought into GS-LC form by first decomposing it into four basic spiders:
\begin{center}
 \splitnode, $\qquad$ \joinnode, $\qquad$ \state{gn}, $\quad$ and $\quad$ \effect{gn},
\end{center}
plus phase shifts and Hadamard nodes (cf.\ Lemma 6 in \cite{backens_zx-calculus_2013}).
The only basic element with no inputs is \state{gn}, therefore any state diagram must contain at least one copy of that (or a cup, which can be rewritten into spiders).
Clearly, \state{gn} is a GS-LC diagram.
Now, for any of the basic spiders, and for any single-qubit Clifford operator, applying it to some GS-LC diagram yields a diagram that can be rewritten to GS-LC form (cf.\ Lemmas 8--12 in \cite{backens_zx-calculus_2013}).
Thus, by induction, any stabilizer state diagram can be brought into GS-LC form (cf.\ Theorem 7 in \cite{backens_zx-calculus_2013}).

\begin{proof}[Proof of Corollary \ref{cor:decompose_scalars}]
 Take any connected component of the scalar diagram which contains more than two nodes.
 Rewrite it into the inner product between some (possibly complicated) state diagram and \effect{gn}.
 This can be done for any connected scalar diagram by decomposing it into basic spiders.
 As \effect{gn} is the only basic spider with no outputs, the scalar must contain at least one copy of \effect{gn} (or a cap, which can be rewritten into spiders).

 The state diagram part, which represents a single-qubit state, can then be brought into GS-LC form.
 Any scalar subdiagrams that ``split off'' the main part of the diagram in this rewriting process  consist of disconnected segments containing at most two nodes each, as can easily be checked by looking at the rewrite rules and the proofs of Lemmas 8--12 in \cite{backens_zx-calculus_2013}.
 Up to scalar factor, the six single-qubit stabilizer states can be written as:
 \[
  \state{gn}, \quad
  \state{gn, label={[gphase]right:$\pi/2$}}, \quad
  \state{gn, label={[gphase]right:$\pi$}}, \quad
  \state{gn, label={[gphase]right:$-\pi/2$}}, \quad 
  \state{rn}, \quad \text{and} \quad 
  \state{rn, label={[rphase]right:$\pi$}}.
 \]
 By the scalar-free stabilizer completeness result, any diagram representing a single-qubit state can be rewritten to one of the above single-node diagrams.
 This node combines with the green effect into a two-node diagram.
\end{proof}

When the scaled rewrite rules are used instead of the scalar-free ones, Corollary \ref{cor:decompose_scalars} remains valid because the scalars appearing in the rewrite rules all consist of disconnected components with at most two nodes each.

\begin{proof}[Proof of Corollary \ref{cor:recognize_zero}]
 A GS-LC diagram in the scalar-free \ZX-calculus cannot be zero as it consists of unitary operators -- controlled-$Z$ gates and single-qubit unitaries -- applied to the state $\state{gn}\ldots\state{gn}$.
 Therefore, a scaled GS-LC diagram is zero if and only if its scalar part is zero.
 Now, for any two scalar diagrams \genscalar{$s$} and \genscalar{$r$},
 \begin{equation}
  \left\llbracket \genscalar{$s$} \; \genscalar{$r$} \right\rrbracket = \left\llbracket \genscalar{$s$} \right\rrbracket \left\llbracket \genscalar{$r$} \right\rrbracket,
 \end{equation}
 i.e.\ putting two scalar diagrams next to each other corresponds to taking the product of their values.
 Thus a scalar diagram is zero if and only if at least one of the disconnected components is zero.
 It is straightforward to check that out of all connected scalar diagrams containing at most two nodes, the diagrams given in \eqref{eq:zero_scalars} are exactly the ones that are zero.
 The result then follows from Corollary \ref{cor:decompose_scalars} and the fact that any stabilizer diagram can be brought into GS-LC form.
\end{proof}

\section{Proofs of theorem and lemmas from Section \ref{s:completeness}}
\label{s:appendix}

\begin{proof}[Proof of Theorem \ref{thm:scalar_nf}]
 By Corollary \ref{cor:decompose_scalars}, we only need to consider diagrams made up of disconnected segments of at most two nodes each.
 Using Lemma \ref{lem:y-states}, any disconnected single node \scalar{gn, label={[gphase]right:$\alpha$}} or \scalar{rn, label={[rphase]right:$\beta$}} can be rewritten into a diagram that is made up of copies of \halfscalar{} and components consisting of exactly one red and one green node.
 Given this, Lemmas \ref{lem:innerprod_wlog} and \ref{lem:overlap_with_ket_zero}, and the fact that we are considering non-zero diagrams only, we can restrict our attention without loss of generality to diagrams built only from \halfscalar{}, \innerprodgr{}, and the following two-node diagrams:
 \begin{equation}\label{eq:basic_overlaps}
  \innerprod{gn,label={[gphase]right:$\pi/2$}}{rn,label={[rphase]right:$\pi/2$}}, \quad
  \innerprod{gn,label={[gphase]right:$\pi/2$}}{rn,label={[rphase]right:$\pi$}}, \quad
  \innerprod{gn,label={[gphase]right:$\pi$}}{rn,label={[rphase]right:$\pi$}}, \quad
  \innerprod{gn,label={[gphase]right:$-\pi/2$}}{rn,label={[rphase]right:$\pi$}}, \; \text{ and } \;
  \innerprod{gn,label={[gphase]right:$-\pi/2$}}{rn,label={[rphase]right:$-\pi/2$}}.
 \end{equation}
 These diagrams can easily be seen to also be in the set \eqref{eq:complex_phases}.
 Thus all that remains to show is that a diagram consisting of several of these elements can be rewritten to a diagram that consists of only one of the diagrams given in \eqref{eq:complex_phases} (or the empty diagram) plus any number of copies of \innerprodgr{} and \halfscalar.
 This rewriting can be done in a step-by-step fashion, so it suffices to look at pairs of diagrams from \eqref{eq:basic_overlaps}.

 The combination of any two diagrams that both contain \state{rn, label={[rphase]right:$\pi$}} can be simplified using Lemma \ref{lem:pi_multiplication}. The products
 \begin{center}
  \innerprod{gn,label={[gphase]right:$\pi/2$}}{rn,label={[rphase]right:$\pi/2$}} \innerprod{gn,label={[gphase]right:$-\pi/2$}}{rn,label={[rphase]right:$-\pi/2$}}, $\qquad$ \innerprod{gn,label={[gphase]right:$\pi/2$}}{rn,label={[rphase]right:$\pi$}} \innerprod{gn,label={[gphase]right:$\pi/2$}}{rn,label={[rphase]right:$\pi/2$}}, $\quad$ and $\quad$ \innerprod{gn,label={[gphase]right:$-\pi/2$}}{rn,label={[rphase]right:$\pi$}} \innerprod{gn,label={[gphase]right:$-\pi/2$}}{rn,label={[rphase]right:$-\pi/2$}}
 \end{center}
 are straightforward as well: the first diagram is shown to be simplifiable in Lemma \ref{lem:omega_inverses}, the latter two are elements of \eqref{eq:complex_phases}, so do not need to be simplified.

 For other combinations, note first that using the spider rule, the $\pi$-copy rule, and the $\pi$-commutation rule, we have:
 \begin{equation}\label{eq:scalar_pi_2_equality}
  \input{tikz_files/multiplication4_der2.0.tikz}.
 \end{equation}
 Thus:
 \begin{equation}\label{eq:omega-dagger_squared}
  \innerprod{gn,label={[gphase]right:$-\pi/2$}}{rn,label={[rphase]right:$-\pi/2$}} \; \innerprod{gn,label={[gphase]right:$-\pi/2$}}{rn,label={[rphase]right:$-\pi/2$}} \;
  = \; \halfscalar \; \innerprod{gn}{rn} \; \innerprod{gn,label={[gphase]right:$-\pi/2$}}{rn,label={[rphase]right:$\pi$}} \; \innerprod{gn,label={[gphase]right:$\pi/2$}}{rn,label={[rphase]right:$\pi/2$}} \; \innerprod{gn,label={[gphase]right:$-\pi/2$}}{rn,label={[rphase]right:$-\pi/2$}} \;
  = \; \innerprod{gn}{rn} \; \innerprod{gn}{rn} \; \innerprod{gn}{rn} \; \innerprod{gn,label={[gphase]right:$-\pi/2$}}{rn,label={[rphase]right:$\pi$}},
 \end{equation}
 where the last equality is by Lemma \ref{lem:omega_inverses}.
 Furthermore, using \eqref{eq:scalar_pi_2_equality} and Lemma \ref{lem:pi_multiplication}, we get:
 \begin{equation}\label{eq:minus_omega}
  \innerprod{gn,label={[gphase]right:$\pi$}}{rn,label={[rphase]right:$\pi$}} \; \innerprod{gn,label={[gphase]right:$\pi/2$}}{rn,label={[rphase]right:$\pi/2$}} \;
  = \; \halfscalar \; \innerprod{gn}{rn} \; \innerprod{gn,label={[gphase]right:$-\pi/2$}}{rn,label={[rphase]right:$\pi$}} \; \innerprod{gn,label={[gphase]right:$-\pi/2$}}{rn,label={[rphase]right:$\pi$}} \; \innerprod{gn,label={[gphase]right:$\pi/2$}}{rn,label={[rphase]right:$\pi/2$}} \;
  = \; \innerprod{gn,label={[gphase]right:$-\pi/2$}}{rn,label={[rphase]right:$\pi$}} \; \innerprod{gn,label={[gphase]right:$-\pi/2$}}{rn,label={[rphase]right:$-\pi/2$}}.
 \end{equation}
 Equalities similar to \eqref{eq:scalar_pi_2_equality}, \eqref{eq:omega-dagger_squared}, and \eqref{eq:minus_omega} can be derived with all the signs flipped.

 This covers all the combinations of two diagrams from \eqref{eq:basic_overlaps}.
 More complicated diagrams can be dealt with step-by-step.
 Once the subdiagrams that involve complex phases are normalised, the real parts of the diagram can be brought into the normal form given in Lemma \ref{lem:modulus_nf}.

 The resulting normal form for complex non-zero scalars can easily be seen to be unique.
\end{proof}

\begin{proof}[Proof of Lemma \ref{lem:innerprod_wlog}]
 The first equality results from the topology meta rule.
 The equality between the second and third diagram follows from the colour change rule and the fact that the Hadamard node is self-inverse:
 \begin{equation}
  \input{tikz_files/colour_change_der2.0.tikz}.
 \end{equation}
 The last equality is again by the topology meta rule.
\end{proof}

\begin{proof}[Proof of Lemma \ref{lem:pi_multiplication}]
 We have:
 \begin{equation}\label{eq:scalar_mult}
  \innerprod{gn, label={[gphase]right:$\alpha$}}{rn, label={[rphase]right:$\pi$}} \; \innerprod{gn, label={[gphase]right:$\beta$}}{rn, label={[rphase]right:$\pi$}} \; = \; \innerprod{gn}{rn} \; \input{tikz_files/multiplication_pi_der.tikz} \; = \; \innerprod{gn}{rn} \; \innerprod{gn, label={[gphase]right:$\alpha+\beta$}}{rn, label={[rphase]right:$\pi$}}
 \end{equation}
 using the copy rule, $\pi$-copy rule, and spider rule.
\end{proof}

\begin{proof}[Proof of Lemma \ref{lem:overlap_with_ket_zero}]
 The case $\alpha=0$ is trivial. For $\alpha=\pi$, note that by the spider and $\pi$-copy rules:
 \begin{equation}\label{eq:pi_remove}
  \input{tikz_files/pi_remove_der2.0.tikz}.
 \end{equation}
 Now, using \eqref{eq:pi_remove}, \eqref{eq:halfscalar_innerprodgr2}, and Lemma \ref{lem:pi_multiplication} with $\beta=-\alpha$, together with various rewrite rules, yields:
 \begin{align}
  \innerprod{gn, label={[gphase]right:$\alpha$}}{rn} \;
  &= \; \halfscalar \;\; \innerprod{gn}{rn} \;\; \innerprod{gn}{rn} \;\; \innerprod{gn, label={[gphase]right:$\alpha$}}{rn} \;
  = \; \halfscalar \;\; \innerprod{gn}{rn} \;\; \innerprod{gn}{rn,label={[rphase]right:$\pi$}} \; \innerprod{gn, label={[gphase]right:$\alpha$}}{rn} \;
  = \; \halfscalar \; \innerprod{gn, label={[gphase]right:$-\alpha$}}{rn, label={[rphase]right:$\pi$}} \; \innerprod{gn, label={[gphase]right:$\alpha$}}{rn, label={[rphase]right:$\pi$}} \; \innerprod{gn, label={[gphase]right:$\alpha$}}{rn} \;
  = \; \halfscalar \; \innerprod{gn, label={[gphase]right:$-\alpha$}}{rn, label={[rphase]right:$\pi$}} \;\; \innerprod{gn}{rn} \;\; \input{tikz_files/overlap_with_ket_zero_1.tikz} \nonumber \\
  &= \; \halfscalar \; \innerprod{gn, label={[gphase]right:$-\alpha$}}{rn, label={[rphase]right:$\pi$}} \;\; \innerprod{gn}{rn} \;\; \input{tikz_files/overlap_with_ket_zero_2.tikz} \;
  = \; \halfscalar \; \innerprod{gn, label={[gphase]right:$-\alpha$}}{rn, label={[rphase]right:$\pi$}} \;\; \innerprod{gn}{rn} \;\; \input{tikz_files/overlap_with_ket_zero_3.tikz} \;
  = \; \halfscalar \; \innerprod{gn, label={[gphase]right:$-\alpha$}}{rn, label={[rphase]right:$\pi$}} \; \innerprod{gn, label={[gphase]right:$\alpha$}}{rn, label={[rphase]right:$\pi$}} \;\; \innerprod{gn}{rn} \;\;
  = \;\; \innerprod{gn}{rn},
 \end{align}
 thus proving the result for any $\alpha$.
\end{proof}

\begin{proof}[Proof of Lemma \ref{lem:y-states}]
 The desired equality can be derived as follows:
 \begin{align}
  \state{rn, label={[rphase]right:$-\pi/2$}} \;
  &= \; \input{tikz_files/y_state_der_1.tikz} \;
  = \; \halfscalar \; \innerprod{gn, label={[gphase]right:$-\pi/2$}}{rn, label={[rphase]right:$-\pi/2$}} \; \input{tikz_files/y_state_der_2.tikz} \;
  = \; \halfscalar \; \innerprod{gn, label={[gphase]right:$-\pi/2$}}{rn, label={[rphase]right:$-\pi/2$}} \; \input{tikz_files/y_state_der_3.tikz} \nonumber \\
  &= \; \halfscalar \; \halfscalar \;\; \innerprod{gn}{rn} \;\; \innerprod{gn}{rn} \;\; \innerprod{gn, label={[gphase]right:$-\pi/2$}}{rn, label={[rphase]right:$-\pi/2$}} \; \input{tikz_files/y_state_der_3.tikz} \;
  = \; \halfscalar \; \halfscalar \;\; \innerprod{gn}{rn} \;\; \innerprod{gn, label={[gphase]right:$-\pi/2$}}{rn, label={[rphase]right:$-\pi/2$}} \; \innerprod{gn}{rn, label={[rphase]right:$\pi/2$}} \; \state{gn, label={[gphase]right:$\pi/2$}} \nonumber \\
  &= \; \halfscalar \; \innerprod{gn, label={[gphase]right:$-\pi/2$}}{rn, label={[rphase]right:$-\pi/2$}} \; \state{gn, label={[gphase]right:$\pi/2$}},
 \end{align}
 using the colour change rule, the Euler decomposition rule, the copy rule, \eqref{eq:halfscalar_innerprodgr2}, and Lemma \ref{lem:overlap_with_ket_zero}.
\end{proof}

\begin{proof}[Proof of Lemma \ref{lem:omega_inverses}]
 By the identity rule, spider rule, Euler decomposition rule, colour change rule, and Lemma \ref{lem:overlap_with_ket_zero}:
 \begin{align}\label{eq:scalar_pi_2_inverse}
  \innerprod{gn,label={[gphase]right:$-\pi/2$}}{rn,label={[rphase]right:$-\pi/2$}} \; \innerprod{gn,label={[gphase]right:$\pi/2$}}{rn,label={[rphase]right:$\pi/2$}}
  &= \; \input{tikz_files/multiplication5_der2.0.tikz} \nonumber \\
  &= \; \innerprod{gn}{rn} \; \innerprod{gn}{rn} \; \innerprod{gn}{rn} \; \innerprod{gn,label={[gphase]right:$-\pi/2$}}{rn} \;
  = \; \innerprod{gn}{rn} \; \innerprod{gn}{rn} \; \innerprod{gn}{rn} \; \innerprod{gn}{rn}.
 \end{align}
 The desired equality follows by multiplying both sides with \halfscalar{} \halfscalar{} and using \eqref{eq:halfscalar_innerprodgr2}.
\end{proof}

\end{document}

%% file: tikzit_header2.0.tex
\usepackage[svgnames]{xcolor} 

\usepackage{tikz}
\usetikzlibrary{shapes.geometric} 
\usetikzlibrary{shapes.misc} 
\usetikzlibrary{shapes.symbols} 
\usetikzlibrary{decorations.pathreplacing} 

\pgfdeclarelayer{edgelayer}
\pgfdeclarelayer{nodelayer}
\pgfsetlayers{edgelayer,nodelayer,main}

\tikzstyle{none}=[inner sep=0pt]
\tikzstyle{rn}=[circle,fill=Red,draw=Black,line width=0.8 pt,minimum size=5pt,inner sep=0pt]
\tikzstyle{gn}=[circle,fill=Lime,draw=Black,line width=0.8 pt,minimum size=5pt,inner sep=0pt]
\tikzstyle{Hadamard}=[rectangle,fill=Yellow,draw=Black,minimum size=8pt,inner sep=0pt,label={center:$\scriptstyle\mathrm{H}$}]
\tikzstyle{gphase}=[rounded rectangle,rounded rectangle arc length=90,fill=Lime!20,inner sep=2pt]
\tikzstyle{rphase}=[rounded rectangle,rounded rectangle arc length=90,fill=Red!20,inner sep=2pt]
\tikzstyle{normalrect}=[rectangle,fill=white,draw=black,minimum height=10pt,minimum width=12pt,inner sep=0pt]

\tikzstyle{bn}=[circle,fill=black,draw=black,inner sep=1pt]
\tikzstyle{bigcloud}=[cloud,fill=white,draw=black]
\tikzstyle{bstar}=[star,fill=Black,draw=Black,minimum size=8pt,inner sep=0pt]

\newcommand{\effect}[1]{\begin{tikzpicture}
	\begin{pgfonlayer}{nodelayer}
		\node [style=none] (0) at (0, -0.25) {};
		\node [style={#1}] (1) at (0, 0.25) {};
	\end{pgfonlayer}
	\begin{pgfonlayer}{edgelayer}
		\draw [] (0.center) to (1);
	\end{pgfonlayer}
\end{tikzpicture}}

\newcommand{\state}[1]{\begin{tikzpicture}
	\begin{pgfonlayer}{nodelayer}
		\node [style={#1}] (0) at (0, 0) {};
		\node [style=none] (1) at (0, 0.5) {};
	\end{pgfonlayer}
	\begin{pgfonlayer}{edgelayer}
		\draw (0) to (1.center);
	\end{pgfonlayer}
\end{tikzpicture}}

\newcommand{\phase}[1]{\begin{tikzpicture}
	\begin{pgfonlayer}{nodelayer}
		\node [style=none] (0) at (0, 0.4) {};
		\node [style={#1}] (1) at (0, -0) {};
		\node [style=none] (2) at (0, -0.4) {};
	\end{pgfonlayer}
	\begin{pgfonlayer}{edgelayer}
		\draw [] (0.center) to (1);
		\draw [] (1) to (2.center);
	\end{pgfonlayer}
\end{tikzpicture}}

\newcommand{\splitnode}[0]{\begin{tikzpicture}
	\begin{pgfonlayer}{nodelayer}
		\node [style=gn] (0) at (0, -0) {};
		\node [style=none] (1) at (0, -0.5) {};
		\node [style=none] (2) at (-0.5, 0.5) {};
		\node [style=none] (3) at (0.5, 0.5) {};
	\end{pgfonlayer}
	\begin{pgfonlayer}{edgelayer}
		\draw [bend right=15] (2.center) to (0);
		\draw [bend left=15] (3.center) to (0);
		\draw (0) to (1.center);
	\end{pgfonlayer}
\end{tikzpicture}}

\newcommand{\joinnode}[0]{\begin{tikzpicture}
	\begin{pgfonlayer}{nodelayer}
		\node [style=gn] (0) at (0, -0) {};
		\node [style=none] (1) at (0, 0.5) {};
		\node [style=none] (2) at (0.5, -0.5) {};
		\node [style=none] (3) at (-0.5, -0.5) {};
	\end{pgfonlayer}
	\begin{pgfonlayer}{edgelayer}
		\draw [bend right=15] (2.center) to (0);
		\draw [bend left=15] (3.center) to (0);
		\draw (0) to (1.center);
	\end{pgfonlayer}
\end{tikzpicture}}

\newcommand{\Hadamard}[0]{\begin{tikzpicture}
	\begin{pgfonlayer}{nodelayer}
		\node [style=Hadamard] (0) at (0, 0) {};
		\node [style=none] (1) at (0, 0.5) {};
		\node [style=none] (2) at (0, -0.5) {};
	\end{pgfonlayer}
	\begin{pgfonlayer}{edgelayer}
		\draw (0.center) to (1);
		\draw (2.center) to (1);
	\end{pgfonlayer}
\end{tikzpicture}}

\newcommand{\scalar}[1]{\begin{tikzpicture}
	\begin{pgfonlayer}{nodelayer}
		\node [style={#1}] (1) at (0, 0) {};
	\end{pgfonlayer}
\end{tikzpicture}}

\newcommand{\innerprod}[2]{\begin{tikzpicture}
	\begin{pgfonlayer}{nodelayer}
		\node [style={#1}] (0) at (0, 0.6) {};
		\node [style={#2}] (1) at (0, -0.6) {};
	\end{pgfonlayer}
	\begin{pgfonlayer}{edgelayer}
		\draw (0) to (1);
	\end{pgfonlayer}
\end{tikzpicture}}

\newcommand{\innerprodgr}[0]{\begin{tikzpicture}
	\begin{pgfonlayer}{nodelayer}
		\node [style=gn] (0) at (0, 0.4) {};
		\node [style=rn] (1) at (0, -0.15) {};
	\end{pgfonlayer}
	\begin{pgfonlayer}{edgelayer}
		\draw (0) to (1);
	\end{pgfonlayer}
\end{tikzpicture}}

\newcommand{\halfscalar}[0]{\begin{tikzpicture}
	\begin{pgfonlayer}{nodelayer}
		\node [style=bstar] (1) at (0, 0) {};
	\end{pgfonlayer}
\end{tikzpicture}}

\newcommand{\genscalar}[1]{\begin{tikzpicture}
	\begin{pgfonlayer}{nodelayer}
		\node [style=diamond,draw=black,fill=white,inner sep=1pt,minimum size=0.4cm] (0) at (0, -0) {#1};
	\end{pgfonlayer}
 \end{tikzpicture}}

\tikzstyle{every picture}=[baseline=-0.1cm,scale=0.5]


%% file: tikz_files/green_spider.tikz
\begin{tikzpicture}[baseline=-0.1cm,decoration=brace]
	\begin{pgfonlayer}{nodelayer}
		\node [style=none] (0) at (-3.6, 1.0) {};
		\node [style=none] (1) at (-3.0, 0.9) {$\ldots$};
		\node [style=none] (2) at (-2.4, 1.0) {};
		\node [style=gn, label={[gphase]right:$\alpha$}] (3) at (-3.0, -0.0) {};
		\node [style=none] (5) at (-3.6, -1.0) {};
		\node [style=none] (6) at (-3.0, -0.9) {$\ldots$};
		\node [style=none] (7) at (-2.4, -1.0) {};
		\node [style=none] (8) at (-3.7, 1.1) {};
		\node [style=none] (9) at (-2.3, 1.1) {};
		\node [style=none] (10) at (-3.7, -1.1) {};
		\node [style=none] (11) at (-2.3, -1.1) {};
		\node [style=none] (12) at (-3.0, 1.6) {$m$};
		\node [style=none] (13) at (-3.0, -1.6) {$n$};
	\end{pgfonlayer}
	\begin{pgfonlayer}{edgelayer}
		\draw [bend left=15] (3) to (7.center);
		\draw [bend right=15] (3) to (5.center);
		\draw [bend right=15] (0.center) to (3);
		\draw [bend right=15] (3) to (2.center);
		\draw [decorate] (8.center) to (9.center);
		\draw [decorate] (11.center) to (10.center);
	\end{pgfonlayer}
\end{tikzpicture}

%% file: tikz_files/red_spider.tikz
\begin{tikzpicture}[baseline=-0.1cm,decoration=brace]
	\begin{pgfonlayer}{nodelayer}
		\node [style=none] (0) at (-3.6, 1.0) {};
		\node [style=none] (1) at (-3.0, 0.9) {$\ldots$};
		\node [style=none] (2) at (-2.4, 1.0) {};
		\node [style=rn,label={[rphase]right:$\beta$}] (3) at (-3.0, -0.0) {};
		\node [style=none] (5) at (-3.6, -1.0) {};
		\node [style=none] (6) at (-3.0, -0.9) {$\ldots$};
		\node [style=none] (7) at (-2.4, -1.0) {};
		\node [style=none] (8) at (-3.7, 1.1) {};
		\node [style=none] (9) at (-2.3, 1.1) {};
		\node [style=none] (10) at (-3.7, -1.1) {};
		\node [style=none] (11) at (-2.3, -1.1) {};
		\node [style=none] (12) at (-3.0, 1.7) {$l$};
		\node [style=none] (13) at (-3.0, -1.7) {$k$};
	\end{pgfonlayer}
	\begin{pgfonlayer}{edgelayer}
		\draw [bend left=15] (3) to (7.center);
		\draw [bend right=15] (3) to (5.center);
		\draw [bend right=15] (0.center) to (3);
		\draw [bend right=15] (3) to (2.center);
		\draw [decorate] (8.center) to (9.center);
		\draw [decorate] (11.center) to (10.center);
	\end{pgfonlayer}
\end{tikzpicture}

%% file: tikz_files/cup_diagram.tikz
\begin{tikzpicture}
	\begin{pgfonlayer}{nodelayer}
		\node [style=none] (0) at (1, 0.25) {};
		\node [style=none] (1) at (2, 0.25) {};
		\node [style=none] (2) at (1.5, -0.25) {};
	\end{pgfonlayer}
	\begin{pgfonlayer}{edgelayer}
		\draw [bend right=45, looseness=1.00] (2.center) to (1.center);
		\draw [bend right=45, looseness=1.00] (0.center) to (2.center);
	\end{pgfonlayer}
\end{tikzpicture}

%% file: tikz_files/scalar_ex.tikz
\begin{tikzpicture}
	\begin{pgfonlayer}{nodelayer}
		\node [style=gn] (0) at (0, -0) {};
		\node [label={[rphase]left:$\pi/2$}, style=rn] (1) at (-0.75, 1) {};
		\node [label={[rphase]right:$\pi/2$}, style=rn] (2) at (0.75, 1) {};
		\node [label={[rphase]left:$\pi/2$}, style=rn] (3) at (-0.75, -1) {};
		\node [label={[rphase]right:$\pi/2$}, style=rn] (4) at (0.75, -1) {};
	\end{pgfonlayer}
	\begin{pgfonlayer}{edgelayer}
		\draw [bend right=15, looseness=1.00] (1) to (0);
		\draw [bend left=15, looseness=1.00] (2) to (0);
		\draw [bend right=15, looseness=1.00] (0) to (3);
		\draw [bend left=15, looseness=1.00] (0) to (4);
	\end{pgfonlayer}
\end{tikzpicture}

%% file: tikz_files/scalar_ex2.tikz
\begin{tikzpicture}
	\begin{pgfonlayer}{nodelayer}
		\node [label={[gphase]right:$-\pi/2$}, style=gn] (0) at (0, 1) {};
		\node [style=Hadamard] (1) at (0, -0) {};
		\node [style=rn] (2) at (0, -1) {};
	\end{pgfonlayer}
	\begin{pgfonlayer}{edgelayer}
		\draw (0) to (2);
	\end{pgfonlayer}
\end{tikzpicture}

%% file: tikz_files/spider2.0.tikz
\begin{tikzpicture}
	\begin{pgfonlayer}{nodelayer}
		\node [style=none] (0) at (-4.5, 1) {};
		\node [style=none] (1) at (-4, 0.9) {$\ldots$};
		\node [style=none] (2) at (-3.5, 1) {};
		\node [style=none] (3) at (-2.5, 1) {};
		\node [style=none] (4) at (-2, 0.9) {$\ldots$};
		\node [style=none] (5) at (-1.5, 1) {};
		\node [style=none] (6) at (0.5, 1) {};
		\node [style=none] (7) at (1.25, 0.9) {$\ldots$};
		\node [style=none] (8) at (2, 1) {};
		\node [style=none] (9) at (2.5, 1) {};
		\node [style=none] (10) at (3.25, 0.9) {$\ldots$};
		\node [style=none] (11) at (4, 1) {};
		\node [style=gn, label={[gphase]left:$\alpha$}] (12) at (-4, 0.25) {};
		\node [style=none] (13) at (0, -0) {$=$};
		\node [style=gn, label={[gphase]right:$\alpha+\beta$}] (14) at (2.25, -0) {};
		\node [style=gn, label={[gphase]right:$\beta$}] (15) at (-2, -0.25) {};
		\node [style=none] (16) at (0.5, -1) {};
		\node [style=none] (17) at (1.25, -0.9) {$\ldots$};
		\node [style=none] (18) at (2, -1) {};
		\node [style=none] (19) at (2.5, -1) {};
		\node [style=none] (20) at (3.25, -0.9) {$\ldots$};
		\node [style=none] (21) at (4, -1) {};
		\node [style=none] (22) at (-4.5, -1) {};
		\node [style=none] (23) at (-4, -0.9) {$\ldots$};
		\node [style=none] (24) at (-3.5, -1) {};
		\node [style=none] (25) at (-2.5, -1) {};
		\node [style=none] (26) at (-2, -0.9) {$\ldots$};
		\node [style=none] (27) at (-1.5, -1) {};
	\end{pgfonlayer}
	\begin{pgfonlayer}{edgelayer}
		\draw [bend right=15, looseness=1.00] (14) to (16.center);
		\draw (9.center) to (14);
		\draw (14) to (19.center);
		\draw [bend left=15, looseness=1.00] (5.center) to (15);
		\draw [bend right=15, looseness=1.00] (12) to (22.center);
		\draw [bend right=15, looseness=1.00] (6.center) to (14);
		\draw [bend left=15, looseness=1.00] (12) to (24.center);
		\draw [bend right=15, looseness=1.00] (0.center) to (12);
		\draw [bend left=15, looseness=1.00] (14) to (21.center);
		\draw [bend left=15, looseness=1.00] (15) to (27.center);
		\draw [bend right=15, looseness=1.00] (12) to (2.center);
		\draw [bend left=15, looseness=1.00] (11.center) to (14);
		\draw [bend right=15, looseness=1.00] (15) to (25.center);
		\draw (14) to (18.center);
		\draw [bend right=15, looseness=1.00] (3.center) to (15);
		\draw (8.center) to (14);
		\draw (12) to (15);
	\end{pgfonlayer}
\end{tikzpicture}

%% file: tikz_files/loop2.0.tikz
\begin{tikzpicture}[baseline=-0.1cm, decoration=brace]
	\begin{pgfonlayer}{nodelayer}
		\node [style=none] (0) at (-2.5, 1) {};
		\node [style=none] (1) at (-2, 0.9) {$\ldots$};
		\node [style=none] (2) at (-1.5, 1) {};
		\node [style=gn, label={[gphase]right:$\alpha$}] (3) at (-2, -0) {};
		\node [style=none] (4) at (-2.5, -1) {};
		\node [style=none] (5) at (-2, -0.9) {$\ldots$};
		\node [style=none] (6) at (-1.5, -1) {};
		\node [style=none] (7) at (-3.25, -0) {};
		\node [style=none] (8) at (-0, -0) {$=$};
		\node [style=none] (9) at (1.25, 0.9) {$\ldots$};
		\node [style=none] (10) at (1.75, 1) {};
		\node [style=none] (11) at (0.75, 1) {};
		\node [style=none] (12) at (1.75, -1) {};
		\node [style=none] (13) at (0.75, -1) {};
		\node [label={[gphase]right:$\alpha$}, style=gn] (14) at (1.25, -0) {};
		\node [style=none] (15) at (1.25, -0.9) {$\ldots$};
		\node [style=none] (16) at (-2.75, 0.5) {};
		\node [style=none] (17) at (-2.75, -0.5) {};
	\end{pgfonlayer}
	\begin{pgfonlayer}{edgelayer}
		\draw [bend left=15, looseness=1.00] (3) to (6.center);
		\draw [bend right=15, looseness=1.00] (3) to (4.center);
		\draw [bend right=15, looseness=1.00] (0.center) to (3);
		\draw [bend right=15, looseness=1.00] (3) to (2.center);
		\draw [bend left=15, looseness=1.00] (14) to (12.center);
		\draw [bend right=15, looseness=1.00] (14) to (13.center);
		\draw [bend right=15, looseness=1.00] (11.center) to (14);
		\draw [bend right=15, looseness=1.00] (14) to (10.center);
		\draw [in=0, out=146, looseness=0.75] (3) to (16.center);
		\draw [bend right=45, looseness=1.00] (16.center) to (7.center);
		\draw [bend right=45, looseness=1.00] (7.center) to (17.center);
		\draw [in=-146, out=0, looseness=0.75] (17.center) to (3);
	\end{pgfonlayer}
\end{tikzpicture}

%% file: tikz_files/cup2.0.tikz
\begin{tikzpicture}
	\begin{pgfonlayer}{nodelayer}
		\node [style=none] (0) at (-2, 0.25) {};
		\node [style=none] (1) at (-1, 0.25) {};
		\node [style=none] (2) at (1, 0.25) {};
		\node [style=none] (3) at (2, 0.25) {};
		\node [style=none] (4) at (0, -0) {$=$};
		\node [style=gn] (5) at (-1.5, -0.25) {};
		\node [style=none] (6) at (1.5, -0.25) {};
	\end{pgfonlayer}
	\begin{pgfonlayer}{edgelayer}
		\draw [bend right=45, looseness=1.00] (0.center) to (5);
		\draw [bend right=45, looseness=1.00] (5) to (1.center);
		\draw [bend right=45, looseness=1.00] (6.center) to (3.center);
		\draw [bend right=45, looseness=1.00] (2.center) to (6.center);
	\end{pgfonlayer}
\end{tikzpicture}

%% file: tikz_files/bialgebra2.0.tikz
\begin{tikzpicture}
	\begin{pgfonlayer}{nodelayer}
		\node [style=none] (0) at (-2.75, 1.25) {};
		\node [style=none] (1) at (-1.25, 1.25) {};
		\node [style=none] (2) at (0.75, 1.25) {};
		\node [style=none] (3) at (1.75, 1.25) {};
		\node [style=gn] (4) at (-2.75, 0.75) {};
		\node [style=gn] (5) at (-1.25, 0.75) {};
		\node [style=rn] (6) at (1.25, 0.5) {};
		\node [style=none] (7) at (0, -0) {$=$};
		\node [style=gn] (8) at (1.25, -0.5) {};
		\node [style=rn] (9) at (-2.75, -0.75) {};
		\node [style=rn] (10) at (-1.25, -0.75) {};
		\node [style=none] (11) at (-2.75, -1.25) {};
		\node [style=none] (12) at (-1.25, -1.25) {};
		\node [style=none] (13) at (0.75, -1.25) {};
		\node [style=none] (14) at (1.75, -1.25) {};
	\end{pgfonlayer}
	\begin{pgfonlayer}{edgelayer}
		\draw (9) to (11.center);
		\draw [bend left=15, looseness=1.00] (3.center) to (6);
		\draw [bend left, looseness=1.00] (5) to (10);
		\draw (10) to (12.center);
		\draw (6) to (8);
		\draw (5) to (9);
		\draw (1.center) to (5);
		\draw [bend right, looseness=1.00] (4) to (9);
		\draw (4) to (10);
		\draw (0.center) to (4);
		\draw [bend right=15, looseness=1.00] (8) to (13.center);
		\draw [bend right=15, looseness=1.00] (2.center) to (6);
		\draw [bend left=15, looseness=1.00] (8) to (14.center);
	\end{pgfonlayer}
\end{tikzpicture}

%% file: tikz_files/copy2.0.tikz
\begin{tikzpicture}
	\begin{pgfonlayer}{nodelayer}
		\node [style=none] (0) at (-1.75, 0.75) {};
		\node [style=rn] (1) at (-1.25, -0) {};
		\node [style=none] (2) at (-0.75, 0.75) {};
		\node [style=gn] (3) at (-1.25, -0.75) {};
		\node [style=none] (4) at (0, -0) {$=$};
		\node [style=none] (5) at (1, 0.5) {};
		\node [style=gn] (6) at (1, -0.25) {};
		\node [style=none] (7) at (2, 0.5) {};
		\node [style=gn] (8) at (2, -0.25) {};
	\end{pgfonlayer}
	\begin{pgfonlayer}{edgelayer}
		\draw [bend right, looseness=1.00] (0.center) to (1);
		\draw (1) to (3.center);
		\draw [bend left, looseness=1.00] (2.center) to (1);
		\draw (5.center) to (6);
		\draw (7.center) to (8);
	\end{pgfonlayer}
\end{tikzpicture}

%% file: tikz_files/colour2.0.tikz
\begin{tikzpicture}
	\begin{pgfonlayer}{nodelayer}
		\node [style=rn, label={[rphase]right:$\alpha$}] (0) at (-2, -0) {};
		\node [style=none] (1) at (-2.75, 1) {};
		\node [style=none] (2) at (-1.25, 1) {};
		\node [style=none] (3) at (-2.75, -1) {};
		\node [style=none] (4) at (-1.25, -1) {};
		\node [style=none] (5) at (0, -0) {$=$};
		\node [style=gn, label={[gphase]right:$\alpha$}] (6) at (2, -0) {};
		\node [style=Hadamard] (7) at (1, 1) {};
		\node [style=Hadamard] (8) at (3, 1) {};
		\node [style=Hadamard] (9) at (1, -1) {};
		\node [style=Hadamard] (10) at (3, -1) {};
		\node [style=none] (11) at (1, 1.5) {};
		\node [style=none] (12) at (3, 1.5) {};
		\node [style=none] (13) at (1, -1.5) {};
		\node [style=none] (14) at (3, -1.5) {};
		\node [style=none] (15) at (2, 1) {$\ldots $};
		\node [style=none] (16) at (2, -1) {$\ldots $};
		\node [style=none] (17) at (-2, 0.9) {$\ldots $};
		\node [style=none] (18) at (-2, -0.9) {$\ldots $};
	\end{pgfonlayer}
	\begin{pgfonlayer}{edgelayer}
		\draw [bend right=15, looseness=1.00] (1.center) to (0);
		\draw [bend right=15, looseness=1.00] (0) to (3.center);
		\draw [bend left=15, looseness=1.00] (2.center) to (0);
		\draw [bend left=15, looseness=1.00] (0) to (4.center);
		\draw (11.center) to (7);
		\draw [bend right=15, looseness=1.00] (7) to (6);
		\draw [bend right=15, looseness=1.00] (6) to (8);
		\draw (8) to (12.center);
		\draw [bend right=15, looseness=1.00] (6) to (9);
		\draw (9) to (13.center);
		\draw (10) to (14.center);
		\draw [bend left=15, looseness=1.00] (6) to (10);
	\end{pgfonlayer}
\end{tikzpicture}

%% file: tikz_files/pi-copy2.0.tikz
\begin{tikzpicture}
	\begin{pgfonlayer}{nodelayer}
		\node [style=none] (0) at (-2.5, 0.75) {};
		\node [style=rn] (1) at (-1.75, -0.25) {};
		\node [style=none] (2) at (-1, 0.75) {};
		\node [style=none] (3) at (-1.75, -1.5) {};
		\node [style=none] (4) at (0, -0) {$=$};
		\node [label={[gphase]right:$\pi$}, style=gn] (5) at (-1.75, -1) {};
		\node [style=none] (6) at (-1.75, 0.65) {$\ldots$};
		\node [style=gn, label={[gphase]left:$\pi$}] (7) at (2, -0) {};
		\node [style=rn] (8) at (2.75, -0.75) {};
		\node [style=none] (9) at (2.75, -1.25) {};
		\node [style=none] (10) at (3.5, 0.5) {};
		\node [style=none] (11) at (2.75, -0) {$\ldots$};
		\node [style=none] (12) at (2, 0.5) {};
		\node [label={[gphase]right:$\pi$}, style=gn] (13) at (3.5, -0) {};
	\end{pgfonlayer}
	\begin{pgfonlayer}{edgelayer}
		\draw [bend right, looseness=1.00] (0.center) to (1);
		\draw (1) to (3.center);
		\draw [bend left, looseness=1.00] (2.center) to (1);
		\draw (8) to (9.center);
		\draw [bend left, looseness=1.00] (8) to (7);
		\draw (7) to (12.center);
		\draw [bend right, looseness=1.00] (8) to (13);
		\draw (13) to (10.center);
	\end{pgfonlayer}
\end{tikzpicture}

%% file: tikz_files/pi-comm2.0.tikz
\begin{tikzpicture}
	\begin{pgfonlayer}{nodelayer}
		\node [style=rn, label={[rphase]right:$\alpha$}] (0) at (-2.25, -0.5) {};
		\node [style=none] (1) at (-2.25, 1) {};
		\node [style=none] (2) at (-2.25, -1) {};
		\node [style=none] (3) at (0, -0) {$=$};
		\node [style=gn, label={[gphase]right:$\pi$}] (4) at (-2.25, 0.5) {};
		\node [label={[rphase]right:$-\alpha$}, style=rn] (5) at (1.25, 0.5) {};
		\node [style=none] (6) at (1.25, 1) {};
		\node [style=none] (7) at (1.25, -1) {};
		\node [label={[gphase]right:$\pi$}, style=gn] (8) at (1.25, -0.5) {};
	\end{pgfonlayer}
	\begin{pgfonlayer}{edgelayer}
		\draw (0) to (2.center);
		\draw (1.center) to (0);
		\draw (5) to (7.center);
		\draw (6.center) to (5);
	\end{pgfonlayer}
\end{tikzpicture}

%% file: tikz_files/Euler_dec2.0.tikz
\begin{tikzpicture}
	\begin{pgfonlayer}{nodelayer}
		\node [style=none] (0) at (-1.25, 0.75) {};
		\node [style=none] (1) at (1.25, 1.75) {};
		\node [style=gn, label={[gphase]right:$\pi/2$}] (2) at (1.25, 1.25) {};
		\node [style=Hadamard] (3) at (-1.25, -0) {};
		\node [style=none] (4) at (0, -0) {$=$};
		\node [style=rn, label={[rphase]right:$\pi/2$}] (5) at (1.25, -0) {};
		\node [style=gn, label={[gphase]right:$\pi/2$}] (6) at (1.25, -1.25) {};
		\node [style=none] (7) at (-1.25, -0.75) {};
		\node [style=none] (8) at (1.25, -1.75) {};
	\end{pgfonlayer}
	\begin{pgfonlayer}{edgelayer}
		\draw (2) to (5);
		\draw (5) to (6);
		\draw (0.center) to (3);
		\draw (6) to (8.center);
		\draw (1.center) to (2);
		\draw (3) to (7.center);
	\end{pgfonlayer}
\end{tikzpicture}

%% file: tikz_files/loop_scalar2.0.tikz
\begin{tikzpicture}
	\begin{pgfonlayer}{nodelayer}
		\node [style=none] (0) at (-0.5, -0) {};
		\node [style=none] (1) at (0, 0.5) {};
		\node [style=none] (2) at (0.5, -0) {};
		\node [style=none] (3) at (0, -0.5) {};
	\end{pgfonlayer}
	\begin{pgfonlayer}{edgelayer}
		\draw [bend left=45, looseness=1.00] (1.center) to (2.center);
		\draw [bend left=45, looseness=1.00] (2.center) to (3.center);
		\draw [bend left=45, looseness=1.00] (3.center) to (0.center);
		\draw [bend left=45, looseness=1.00] (0.center) to (1.center);
	\end{pgfonlayer}
\end{tikzpicture}

%% file: tikz_files/loop_gn.tikz
\begin{tikzpicture}
	\begin{pgfonlayer}{nodelayer}
		\node [style=none] (0) at (-0.5, -0) {};
		\node [style=none] (1) at (0, 0.5) {};
		\node [style=none] (2) at (0.5, -0) {};
		\node [style=gn] (3) at (0, -0.5) {};
	\end{pgfonlayer}
	\begin{pgfonlayer}{edgelayer}
		\draw [bend left=45, looseness=1.00] (1.center) to (2.center);
		\draw [bend left=45, looseness=1.00] (2.center) to (3);
		\draw [bend left=45, looseness=1.00] (3) to (0.center);
		\draw [bend left=45, looseness=1.00] (0.center) to (1.center);
	\end{pgfonlayer}
\end{tikzpicture}

%% file: tikz_files/bialgebra_scalars2.0.tikz
\begin{tikzpicture}
	\begin{pgfonlayer}{nodelayer}
		\node [style=none] (0) at (-2.75, 1.25) {};
		\node [style=none] (1) at (-1.25, 1.25) {};
		\node [style=none] (2) at (0.75, 1.25) {};
		\node [style=none] (3) at (1.75, 1.25) {};
		\node [style=gn] (4) at (-2.75, 0.75) {};
		\node [style=gn] (5) at (-1.25, 0.75) {};
		\node [style=rn] (6) at (1.25, 0.5) {};
		\node [style=none] (7) at (0, -0) {$=$};
		\node [style=gn] (8) at (1.25, -0.5) {};
		\node [style=rn] (9) at (-2.75, -0.75) {};
		\node [style=rn] (10) at (-1.25, -0.75) {};
		\node [style=none] (11) at (-2.75, -1.25) {};
		\node [style=none] (12) at (-1.25, -1.25) {};
		\node [style=none] (13) at (0.75, -1.25) {};
		\node [style=none] (14) at (1.75, -1.25) {};
		\node [style=gn] (15) at (-4, 0.5) {};
		\node [style=rn] (16) at (-4, -0.5) {};
	\end{pgfonlayer}
	\begin{pgfonlayer}{edgelayer}
		\draw (9) to (11.center);
		\draw [bend left=15, looseness=1.00] (3.center) to (6);
		\draw [bend left, looseness=1.00] (5) to (10);
		\draw (10) to (12.center);
		\draw (6) to (8);
		\draw (5) to (9);
		\draw (1.center) to (5);
		\draw [bend right, looseness=1.00] (4) to (9);
		\draw (4) to (10);
		\draw (0.center) to (4);
		\draw [bend right=15, looseness=1.00] (8) to (13.center);
		\draw [bend right=15, looseness=1.00] (2.center) to (6);
		\draw [bend left=15, looseness=1.00] (8) to (14.center);
		\draw (15) to (16);
	\end{pgfonlayer}
\end{tikzpicture}

%% file: tikz_files/copy_scalars2.0.tikz
\begin{tikzpicture}
	\begin{pgfonlayer}{nodelayer}
		\node [style=none] (0) at (-1.75, 0.75) {};
		\node [style=rn] (1) at (-1.25, -0) {};
		\node [style=none] (2) at (-0.75, 0.75) {};
		\node [style=gn] (3) at (-1.25, -0.75) {};
		\node [style=none] (4) at (0, -0) {$=$};
		\node [style=none] (5) at (1, 0.5) {};
		\node [style=gn] (6) at (1, -0.25) {};
		\node [style=none] (7) at (2, 0.5) {};
		\node [style=gn] (8) at (2, -0.25) {};
		\node [style=gn] (9) at (-2.75, 0.5) {};
		\node [style=rn] (10) at (-2.75, -0.5) {};
	\end{pgfonlayer}
	\begin{pgfonlayer}{edgelayer}
		\draw [bend right, looseness=1.00] (0.center) to (1);
		\draw (1) to (3.center);
		\draw [bend left, looseness=1.00] (2.center) to (1);
		\draw (5.center) to (6);
		\draw (7.center) to (8);
		\draw (9) to (10);
	\end{pgfonlayer}
\end{tikzpicture}

%% file: tikz_files/Hopf_scalars2.0.tikz
\begin{tikzpicture}
	\begin{pgfonlayer}{nodelayer}
		\node [style=gn] (0) at (-1.5, 0.75) {};
		\node [style=none] (1) at (-2, -0) {};
		\node [style=rn] (2) at (-1.5, -0.75) {};
		\node [style=none] (3) at (-1, -0) {};
		\node [style=none] (4) at (-1.5, 1.25) {};
		\node [style=none] (5) at (-1.5, -1.25) {};
		\node [style=none] (6) at (0, -0) {$=$};
		\node [style=none] (7) at (1, 1.25) {};
		\node [style=none] (8) at (1, -1.25) {};
		\node [style=gn] (9) at (1, 0.75) {};
		\node [style=rn] (10) at (1, -0.75) {};
		\node [style=gn] (11) at (-3, 0.5) {};
		\node [style=gn] (12) at (-3.75, 0.5) {};
		\node [style=rn] (13) at (-3, -0.5) {};
		\node [style=rn] (14) at (-3.75, -0.5) {};
	\end{pgfonlayer}
	\begin{pgfonlayer}{edgelayer}
		\draw (4.center) to (0);
		\draw [bend right, looseness=1.00] (0) to (1.center);
		\draw [bend right, looseness=1.00] (1.center) to (2);
		\draw (2) to (5.center);
		\draw [bend left, looseness=1.00] (0) to (3.center);
		\draw [bend left, looseness=1.00] (3.center) to (2);
		\draw (7.center) to (9);
		\draw (10) to (8.center);
		\draw (11) to (13);
		\draw (12) to (14);
	\end{pgfonlayer}
\end{tikzpicture}

%% file: tikz_files/pi-comm_scalars2.0.tikz
\begin{tikzpicture}
	\begin{pgfonlayer}{nodelayer}
		\node [label={[rphase]right:$\alpha$}, style=rn] (0) at (-2.25, -0.5) {};
		\node [style=none] (1) at (-2.25, 1) {};
		\node [style=none] (2) at (-2.25, -1) {};
		\node [style=none] (3) at (0, -0) {$=$};
		\node [label={[gphase]right:$\pi$}, style=gn] (4) at (-2.25, 0.5) {};
		\node [style=rn, label={[rphase]right:$-\alpha$}] (5) at (3.25, 0.5) {};
		\node [style=none] (6) at (3.25, 1) {};
		\node [style=none] (7) at (3.25, -1) {};
		\node [style=gn, label={[gphase]right:$\pi$}] (8) at (3.25, -0.5) {};
		\node [style=gn] (9) at (-3.25, 0.5) {};
		\node [style=rn] (10) at (-3.25, -0.5) {};
		\node [label={[rphase]right:$\alpha$}, style=rn] (11) at (1.25, -0.5) {};
		\node [label={[gphase]right:$\pi$}, style=gn] (12) at (1.25, 0.5) {};
	\end{pgfonlayer}
	\begin{pgfonlayer}{edgelayer}
		\draw (0) to (2.center);
		\draw (1.center) to (0);
		\draw (5) to (7.center);
		\draw (6.center) to (5);
		\draw (9) to (10);
		\draw (12) to (11);
	\end{pgfonlayer}
\end{tikzpicture}

%% file: tikz_files/Euler_dec_scalars2.0.tikz
\begin{tikzpicture}
	\begin{pgfonlayer}{nodelayer}
		\node [style=none] (0) at (-1.25, 0.75) {};
		\node [style=none] (1) at (4.25, 1.75) {};
		\node [style=gn, label={[gphase]right:$\pi/2$}] (2) at (4.25, 1.25) {};
		\node [style=Hadamard] (3) at (-1.25, -0) {};
		\node [style=none] (4) at (0, -0) {$=$};
		\node [style=rn, label={[rphase]right:$\pi/2$}] (5) at (4.25, -0) {};
		\node [style=gn, label={[gphase]right:$\pi/2$}] (6) at (4.25, -1.25) {};
		\node [style=none] (7) at (-1.25, -0.75) {};
		\node [style=none] (8) at (4.25, -1.75) {};
		\node [style=gn] (9) at (-2.25, 0.5) {};
		\node [style=gn] (10) at (-3, 0.5) {};
		\node [style=rn] (11) at (-2.25, -0.5) {};
		\node [style=rn] (12) at (-3, -0.5) {};
		\node [style=gn, label={[gphase]right:$-\pi/2$}] (13) at (1, 0.625) {};
		\node [style=rn, label={[rphase]right:$-\pi/2$}] (14) at (1, -0.625) {};
	\end{pgfonlayer}
	\begin{pgfonlayer}{edgelayer}
		\draw (2) to (5);
		\draw (5) to (6);
		\draw (0.center) to (3);
		\draw (6) to (8.center);
		\draw (1.center) to (2);
		\draw (3) to (7.center);
		\draw (10) to (12);
		\draw (9) to (11);
		\draw (13) to (14);
	\end{pgfonlayer}
\end{tikzpicture}

%% file: tikz_files/copy_example.tikz
\begin{tikzpicture}
	\begin{pgfonlayer}{nodelayer}
		\node [style=none] (0) at (-1.75, 0.75) {};
		\node [style=rn] (1) at (-1.25, -0) {};
		\node [style=none] (2) at (-0.75, 0.75) {};
		\node [style=gn] (3) at (-1.25, -0.75) {};
	\end{pgfonlayer}
	\begin{pgfonlayer}{edgelayer}
		\draw [bend right, looseness=1.00] (0.center) to (1);
		\draw (1) to (3.center);
		\draw [bend left, looseness=1.00] (2.center) to (1);
	\end{pgfonlayer}
\end{tikzpicture}

%% file: tikz_files/inverse_der4.tikz
\begin{tikzpicture}
	\begin{pgfonlayer}{nodelayer}
		\node [style=none] (0) at (3.25, -0.5) {};
		\node [style=none] (1) at (3.75, -0) {};
		\node [style=none] (2) at (2.75, -0) {};
		\node [style=gn] (3) at (3.25, 0.5) {};
	\end{pgfonlayer}
	\begin{pgfonlayer}{edgelayer}
		\draw [bend right=45, looseness=1.00] (3) to (2.center);
		\draw [bend right=45, looseness=1.00] (2.center) to (0.center);
		\draw [bend left=45, looseness=1.00] (3) to (1.center);
		\draw [bend left=45, looseness=1.00] (1.center) to (0.center);
	\end{pgfonlayer}
\end{tikzpicture}

%% file: tikz_files/inverse_der3.tikz
\begin{tikzpicture}
	\begin{pgfonlayer}{nodelayer}
		\node [style=rn] (0) at (3.25, -0.5) {};
		\node [style=none] (1) at (3.75, -0) {};
		\node [style=none] (2) at (2.75, -0) {};
		\node [style=gn] (3) at (3.25, 0.5) {};
	\end{pgfonlayer}
	\begin{pgfonlayer}{edgelayer}
		\draw [bend right=45, looseness=1.00] (3) to (2.center);
		\draw [bend right=45, looseness=1.00] (2.center) to (0);
		\draw [bend left=45, looseness=1.00] (3) to (1.center);
		\draw [bend left=45, looseness=1.00] (1.center) to (0);
	\end{pgfonlayer}
\end{tikzpicture}

%% file: tikz_files/inverse_der2.tikz
\begin{tikzpicture}
	\begin{pgfonlayer}{nodelayer}
		\node [style=rn] (0) at (3.25, -0.75) {};
		\node [style=gn] (1) at (3.25, 1.5) {};
		\node [style=rn] (2) at (3.25, -1.5) {};
		\node [style=none] (3) at (3.75, -0) {};
		\node [style=none] (4) at (2.75, -0) {};
		\node [style=gn] (5) at (3.25, 0.75) {};
	\end{pgfonlayer}
	\begin{pgfonlayer}{edgelayer}
		\draw (1.center) to (5);
		\draw [bend right, looseness=1.00] (5) to (4.center);
		\draw [bend right, looseness=1.00] (4.center) to (0);
		\draw (0) to (2.center);
		\draw [bend left, looseness=1.00] (5) to (3.center);
		\draw [bend left, looseness=1.00] (3.center) to (0);
	\end{pgfonlayer}
\end{tikzpicture}

%% file: tikz_files/inverse_der1.tikz
\begin{tikzpicture}
	\begin{pgfonlayer}{nodelayer}
		\node [style=gn] (0) at (-1, 0.5) {};
		\node [style=gn] (1) at (-1, 1.25) {};
		\node [style=rn] (2) at (-1, -0.5) {};
		\node [style=rn] (3) at (-1, -1.25) {};
	\end{pgfonlayer}
	\begin{pgfonlayer}{edgelayer}
		\draw (1) to (0);
		\draw (2) to (3);
	\end{pgfonlayer}
\end{tikzpicture}

%% file: tikz_files/zero2.0.tikz
\begin{tikzpicture}
	\begin{pgfonlayer}{nodelayer}
		\node [style=none] (0) at (-1.25, 1.25) {};
		\node [style=none] (1) at (-1.25, -1.25) {};
		\node [style=none] (2) at (0, -0) {$=$};
		\node [style=gn] (3) at (3, 0.5) {};
		\node [style=none] (4) at (3, 1.25) {};
		\node [style=none] (5) at (3, -1.25) {};
		\node [style=rn] (6) at (3, -0.5) {};
		\node [style=gn, label={[gphase]right:$\pi$}] (7) at (-3, -0) {};
		\node [label={[gphase]right:$\pi$}, style=gn] (8) at (1.25, -0) {};
	\end{pgfonlayer}
	\begin{pgfonlayer}{edgelayer}
		\draw (4.center) to (3);
		\draw (0.center) to (1.center);
		\draw (6) to (5.center);
	\end{pgfonlayer}
\end{tikzpicture}

%% file: tikz_files/zero_normal_form.tikz
\begin{tikzpicture}[baseline=-0.1cm,decoration=brace]
	\begin{pgfonlayer}{nodelayer}
		\node [label={[gphase]right:$\pi$}, style=gn] (0) at (-2, -0) {};
		\node [style=gn] (1) at (-0.2, -0.5) {};
		\node [style=gn] (2) at (0, 0.5) {};
		\node [style=gn] (3) at (2, 0.5) {};
		\node [style=gn] (4) at (2.2, -0.5) {};
		\node [style=none] (5) at (0, 1) {};
		\node [style=none] (6) at (2, 1) {};
		\node [style=none] (7) at (-0.2, -1) {};
		\node [style=none] (8) at (2.2, -1) {};
		\node [style=none] (9) at (1, 0.5) {$\ldots$};
		\node [style=none] (10) at (1, -0.5) {$\ldots$};
		\node [style=none] (11) at (-0.2, 1.2) {};
		\node [style=none] (12) at (2.2, 1.2) {};
		\node [style=none] (13) at (-0.4, -1.2) {};
		\node [style=none] (14) at (2.4, -1.2) {};
		\node [style=none] (15) at (1, 1.7) {$m$};
		\node [style=none] (16) at (1, -1.7) {$n$};
	\end{pgfonlayer}
	\begin{pgfonlayer}{edgelayer}
		\draw (5.center) to (2);
		\draw (1) to (7.center);
		\draw (6.center) to (3);
		\draw (4) to (8.center);
		\draw [decorate] (11.center) to (12.center);
		\draw [decorate] (14.center) to (13.center);
	\end{pgfonlayer}
\end{tikzpicture}

%% file: tikz_files/bell1.tikz
\begin{tikzpicture}
	\begin{pgfonlayer}{nodelayer}
		\node [style=none] (0) at (-0.5, -0) {};
		\node [style=none] (1) at (0.5, -0) {};
		\node [style=none] (2) at (0, -0.5) {};
		\node [style=rn] (3) at (-0.5, 0.25) {};
		\node [style=rn] (4) at (0.5, 0.25) {};
	\end{pgfonlayer}
	\begin{pgfonlayer}{edgelayer}
		\draw (3) to (0.center);
		\draw [bend right=45, looseness=1.00] (0.center) to (2.center);
		\draw [bend right=45, looseness=1.00] (2.center) to (1.center);
		\draw (1.center) to (4);
	\end{pgfonlayer}
\end{tikzpicture}

%% file: tikz_files/innerprodrg.tikz
\begin{tikzpicture}
	\begin{pgfonlayer}{nodelayer}
		\node [style=rn] (0) at (0, 0.4) {};
		\node [style=gn] (1) at (0, -0.15) {};
	\end{pgfonlayer}
	\begin{pgfonlayer}{edgelayer}
		\draw (0) to (1);
	\end{pgfonlayer}
\end{tikzpicture}

%% file: tikz_files/bell5.tikz
\begin{tikzpicture}
	\begin{pgfonlayer}{nodelayer}
		\node [style=none] (0) at (-0.5, -0) {};
		\node [style=none] (1) at (0.5, -0) {};
		\node [style=none] (2) at (0, -0.5) {};
		\node [label={[rphase]left:$\pi$}, style=rn] (3) at (-0.5, 0.25) {};
		\node [label={[rphase]right:$\pi$}, style=rn] (4) at (0.5, 0.25) {};
	\end{pgfonlayer}
	\begin{pgfonlayer}{edgelayer}
		\draw (3) to (0.center);
		\draw [bend right=45, looseness=1.00] (0.center) to (2.center);
		\draw [bend right=45, looseness=1.00] (2.center) to (1.center);
		\draw (1.center) to (4);
	\end{pgfonlayer}
\end{tikzpicture}

%% file: tikz_files/bell2.tikz
\begin{tikzpicture}
	\begin{pgfonlayer}{nodelayer}
		\node [style=none] (0) at (-0.5, -0) {};
		\node [style=none] (1) at (0.5, -0) {};
		\node [style=none] (2) at (0, -0.5) {};
		\node [style=rn] (3) at (-0.5, 0.25) {};
		\node [label={[rphase]right:$\pi$}, style=rn] (4) at (0.5, 0.25) {};
	\end{pgfonlayer}
	\begin{pgfonlayer}{edgelayer}
		\draw (3) to (0.center);
		\draw [bend right=45, looseness=1.00] (0.center) to (2.center);
		\draw [bend right=45, looseness=1.00] (2.center) to (1.center);
		\draw (1.center) to (4);
	\end{pgfonlayer}
\end{tikzpicture}

%% file: tikz_files/bell3.tikz
\begin{tikzpicture}
	\begin{pgfonlayer}{nodelayer}
		\node [style=none] (0) at (-0.5, 0.75) {};
		\node [style=none] (1) at (0.5, 0.75) {};
		\node [style=none] (2) at (0, 0.25) {};
		\node [style=rn, label={[rphase]left:$-\theta$}] (3) at (-0.5, 1) {};
		\node [style=gn, label={[gphase]right:$-\phi$}] (4) at (0.5, 1) {};
		\node [style=none] (5) at (0.5, -0.75) {};
		\node [style=rn, label={[rphase]left:$\theta$}] (6) at (-0.5, -1) {};
		\node [style=none] (7) at (-0.5, -0.75) {};
		\node [style=none] (8) at (0, -0.25) {};
		\node [style=gn, label={[gphase]right:$\phi$}] (9) at (0.5, -1) {};
	\end{pgfonlayer}
	\begin{pgfonlayer}{edgelayer}
		\draw (3) to (0.center);
		\draw [bend right=45, looseness=1.00] (0.center) to (2.center);
		\draw [bend right=45, looseness=1.00] (2.center) to (1.center);
		\draw (1.center) to (4);
		\draw (9) to (5.center);
		\draw [bend right=45, looseness=1.00] (5.center) to (8.center);
		\draw [bend right=45, looseness=1.00] (8.center) to (7.center);
		\draw (7.center) to (6);
	\end{pgfonlayer}
\end{tikzpicture}

%% file: tikz_files/example_GS-LC.tikz
\begin{tikzpicture}[baseline=-0.1cm]
	\begin{pgfonlayer}{nodelayer}
		\node [style=gn] (0) at (4.5, -2) {};
		\node [style=gn] (1) at (1.5, -0) {};
		\node [style=gn] (2) at (-1.5, -0) {};
		\node [style=gn] (3) at (-4.5, -2) {};
		\node [style=none] (4) at (-4.5, 2.75) {};
		\node [style=none] (5) at (-1.5, 2.75) {};
		\node [style=none] (6) at (1.5, 2.75) {};
		\node [style=none] (7) at (4.5, 2.75) {};
		\node [style=Hadamard] (8) at (0, -0) {};
		\node [style=Hadamard] (9) at (-1.5, -1) {};
		\node [style=Hadamard] (10) at (0, -2) {};
		\node [style=Hadamard] (11) at (1.5, -1) {};
		\node [style=rn, label={[rphase]right:$\pi/2$}] (12) at (1.5, 1) {};
		\node [style=gn, label={[gphase]right:$-\pi/2$}] (13) at (-1.5, 1) {};
		\node [style=gn, label={[gphase]right:$-\pi/2$}] (14) at (-4.5, 1) {};
		\node [style=Hadamard] (15) at (-3, -1) {};
		\node [style=Hadamard] (16) at (-1.5, 2) {};
		\node [style=rn, label={[rphase]right:$\pi$}] (17) at (-4.5, 2) {};
	\end{pgfonlayer}
	\begin{pgfonlayer}{edgelayer}
		\draw (2) to (1);
		\draw (0) to (2);
		\draw (3) to (1);
		\draw (3) to (0);
		\draw (4.center) to (3);
		\draw (5.center) to (2);
		\draw (6.center) to (1);
		\draw (7.center) to (0);
		\draw (2) to (3);
	\end{pgfonlayer}
\end{tikzpicture}

%% file: tikz_files/multiplication4_der2.0.tikz
\begin{tikzpicture}
	\begin{pgfonlayer}{nodelayer}
		\node [label={[rphase]right:$\pi$}, style=rn] (0) at (-2.75, 0.5) {};
		\node [style=none] (1) at (0.5, -0) {$=$};
		\node [label={[gphase]right:$-\pi/2$}, style=gn] (2) at (-2.75, -0.5) {};
		\node [style=gn] (3) at (-3.75, 0.5) {};
		\node [style=rn] (4) at (-3.75, -0.5) {};
		\node [label={[rphase]right:$\pi$}, style=rn] (5) at (1.75, -0.5) {};
		\node [label={[gphase]right:$-\pi/2$}, style=gn] (6) at (1.75, 0.5) {};
		\node [style=gn] (7) at (-8.75, 0.5) {};
		\node [style=rn] (8) at (-8.75, -0.5) {};
		\node [label={[gphase]right:$-\pi/2$}, style=gn] (9) at (-8, -0) {};
		\node [label={[rphase]right:$-\pi/2$}, style=rn] (10) at (-8, -1.25) {};
		\node [style=none] (11) at (-4.75, -0) {$=$};
		\node [style=rn] (12) at (-13.75, -0.625) {};
		\node [style=rn, label={[rphase]right:$-\pi/2$}] (13) at (-13, -0.625) {};
		\node [style=gn, label={[gphase]right:$-\pi/2$}] (14) at (-13, 0.625) {};
		\node [style=gn] (15) at (-13.75, 0.625) {};
		\node [style=none] (16) at (-9.75, -0) {$=$};
		\node [style=gn] (17) at (-8, 1.25) {};
		\node [style=gn] (18) at (-2.75, 1.5) {};
		\node [style=rn, label={[rphase]right:$-\pi/2$}] (19) at (-2.75, -1.75) {};
		\node [style=gn, label={[gphase]right:$\pi/2$}] (20) at (5, 0.5) {};
		\node [style=rn, label={[rphase]right:$\pi$}] (21) at (5, -0.5) {};
		\node [label={[rphase]right:$-\pi/2$}, style=rn] (22) at (5, -1.5) {};
		\node [style=gn] (23) at (5, 1.5) {};
		\node [style=none] (24) at (7.5, -0) {$=$};
		\node [style=rn, label={[rphase]right:$\pi$}] (25) at (8.75, -0.625) {};
		\node [style=rn, label={[rphase]right:$\pi/2$}] (26) at (12, -0.625) {};
		\node [style=gn, label={[gphase]right:$-\pi/2$}] (27) at (8.75, 0.625) {};
		\node [style=gn, label={[gphase]right:$\pi/2$}] (28) at (12, 0.625) {};
	\end{pgfonlayer}
	\begin{pgfonlayer}{edgelayer}
		\draw (3) to (4);
		\draw (6) to (5);
		\draw (7) to (8);
		\draw (15) to (12);
		\draw (14) to (13);
		\draw (17) to (10);
		\draw (18) to (19);
		\draw (23) to (22);
		\draw (27) to (25);
		\draw (28) to (26);
	\end{pgfonlayer}
\end{tikzpicture}

%% file: tikz_files/colour_change_der2.0.tikz
\begin{tikzpicture}
	\begin{pgfonlayer}{nodelayer}
		\node [label={[gphase]right:$\alpha$}, style=gn] (0) at (-3.5, 0.625) {};
		\node [label={[rphase]right:$\beta$}, style=rn] (1) at (-3.5, -0.625) {};
		\node [style=none] (2) at (-1.25, -0) {$=$};
		\node [style=gn, label={[gphase]right:$\beta$}] (3) at (0, -1.25) {};
		\node [style=none] (4) at (2, -0) {$=$};
		\node [style=rn, label={[rphase]right:$\alpha$}] (5) at (0, 1.25) {};
		\node [style=Hadamard] (6) at (0, 0.5) {};
		\node [style=Hadamard] (7) at (0, -0.5) {};
		\node [style=rn, label={[rphase]right:$\alpha$}] (8) at (3, 0.625) {};
		\node [style=gn, label={[gphase]right:$\beta$}] (9) at (3, -0.625) {};
	\end{pgfonlayer}
	\begin{pgfonlayer}{edgelayer}
		\draw (0) to (1);
		\draw (5) to (3);
		\draw (8) to (9);
	\end{pgfonlayer}
\end{tikzpicture}

%% file: tikz_files/multiplication_pi_der.tikz
\begin{tikzpicture}
	\begin{pgfonlayer}{nodelayer}
		\node [style=gn] (0) at (0, -0) {};
		\node [style=gn, label={[gphase]left:$\alpha$}] (1) at (-0.5, 1) {};
		\node [style=gn, label={[gphase]right:$\beta$}] (2) at (0.5, 1) {};
		\node [style=rn, label={[rphase]right:$\pi$}] (3) at (0, -1) {};
	\end{pgfonlayer}
	\begin{pgfonlayer}{edgelayer}
		\draw [bend right, looseness=1.00] (1) to (0);
		\draw [bend left, looseness=1.00] (2) to (0);
		\draw (0) to (3);
	\end{pgfonlayer}
\end{tikzpicture}

%% file: tikz_files/pi_remove_der2.0.tikz
\begin{tikzpicture}
	\begin{pgfonlayer}{nodelayer}
		\node [style=gn, label={[gphase]right:$\pi$}] (0) at (-3.25, 0.5) {};
		\node [style=rn] (1) at (-3.25, -0.5) {};
		\node [style=none] (2) at (-1.25, -0) {$=$};
		\node [style=rn] (3) at (0, -1) {};
		\node [style=none] (4) at (2, -0) {$=$};
		\node [style=gn] (5) at (0, 1) {};
		\node [style=rn] (6) at (3, -0.5) {};
		\node [style=gn] (7) at (3, 0.5) {};
		\node [style=gn, label={[gphase]right:$\pi$}] (8) at (0, -0) {};
	\end{pgfonlayer}
	\begin{pgfonlayer}{edgelayer}
		\draw (0) to (1);
		\draw (5) to (3);
		\draw (6) to (7);
	\end{pgfonlayer}
\end{tikzpicture}

%% file: tikz_files/overlap_with_ket_zero_1.tikz
\begin{tikzpicture}
	\begin{pgfonlayer}{nodelayer}
		\node [style=gn] (0) at (0, -0.5) {};
		\node [label={[rphase]left:$\pi$}, style=rn] (1) at (-0.5, 0.5) {};
		\node [style=rn] (2) at (0, -1.5) {};
		\node [label={[gphase]right:$\alpha$}, style=gn] (3) at (0.5, 0.5) {};
		\node [label={[gphase]left:$\alpha$}, style=gn] (4) at (-0.5, 1.5) {};
	\end{pgfonlayer}
	\begin{pgfonlayer}{edgelayer}
		\draw [bend left, looseness=1.00] (3) to (0);
		\draw (0) to (2);
		\draw (4) to (1);
		\draw [bend right, looseness=1.00] (1) to (0);
	\end{pgfonlayer}
\end{tikzpicture}

%% file: tikz_files/overlap_with_ket_zero_2.tikz
\begin{tikzpicture}
	\begin{pgfonlayer}{nodelayer}
		\node [style=gn, label={[gphase]left:$\alpha$}] (0) at (-0.5, 0.5) {};
		\node [style=rn, label={[rphase]right:$\pi$}] (1) at (0, -1.5) {};
		\node [style=gn] (2) at (0, -0.5) {};
		\node [style=gn, label={[gphase]right:$\alpha$}] (3) at (0.5, 1.5) {};
		\node [style=rn, label={[rphase]right:$\pi$}] (4) at (0.5, 0.5) {};
	\end{pgfonlayer}
	\begin{pgfonlayer}{edgelayer}
		\draw (2) to (1);
		\draw [bend right, looseness=1.00] (0) to (2);
		\draw [bend left, looseness=1.00] (4) to (2);
		\draw (3) to (4);
	\end{pgfonlayer}
\end{tikzpicture}

%% file: tikz_files/overlap_with_ket_zero_3.tikz
\begin{tikzpicture}
	\begin{pgfonlayer}{nodelayer}
		\node [style=rn, label={[rphase]right:$\pi$}] (0) at (0, -1.5) {};
		\node [label={[gphase]right:$\alpha$}, style=gn] (1) at (0, 1.5) {};
		\node [style=gn, label={[gphase]right:$\alpha$}] (2) at (0, -0.5) {};
		\node [style=rn, label={[rphase]right:$\pi$}] (3) at (0, 0.5) {};
	\end{pgfonlayer}
	\begin{pgfonlayer}{edgelayer}
		\draw (1) to (0);
	\end{pgfonlayer}
\end{tikzpicture}

%% file: tikz_files/y_state_der_1.tikz
\begin{tikzpicture}
	\begin{pgfonlayer}{nodelayer}
		\node [style=Hadamard] (0) at (0, -0) {};
		\node [style=none] (1) at (0, 0.5) {};
		\node [label={[gphase]right:$-\pi/2$}, style=gn] (2) at (0, -0.75) {};
	\end{pgfonlayer}
	\begin{pgfonlayer}{edgelayer}
		\draw (1.center) to (2);
	\end{pgfonlayer}
\end{tikzpicture}

%% file: tikz_files/y_state_der_2.tikz
\begin{tikzpicture}
	\begin{pgfonlayer}{nodelayer}
		\node [label={[rphase]right:$\pi/2$}, style=rn] (0) at (0, -0.25) {};
		\node [label={[gphase]right:$\pi/2$}, style=gn] (1) at (0, 1) {};
		\node [style=gn] (2) at (0, -1.5) {};
		\node [style=none] (3) at (0, 1.5) {};
	\end{pgfonlayer}
	\begin{pgfonlayer}{edgelayer}
		\draw (3.center) to (2);
	\end{pgfonlayer}
\end{tikzpicture}

%% file: tikz_files/y_state_der_3.tikz
\begin{tikzpicture}
	\begin{pgfonlayer}{nodelayer}
		\node [style=gn] (0) at (0, -1.25) {};
		\node [style=none] (1) at (-0.5, 1.25) {};
		\node [style=rn, label={[rphase]right:$\pi/2$}] (2) at (0.5, 0.75) {};
		\node [style=rn] (3) at (0, -0.25) {};
		\node [style=gn, label={[gphase]left:$\pi/2$}] (4) at (-0.5, 0.75) {};
	\end{pgfonlayer}
	\begin{pgfonlayer}{edgelayer}
		\draw (1.center) to (4);
		\draw [bend right, looseness=1.00] (4) to (3);
		\draw [bend left, looseness=1.00] (2) to (3);
		\draw (3) to (0);
	\end{pgfonlayer}
\end{tikzpicture}

%% file: tikz_files/multiplication5_der2.0.tikz
\begin{tikzpicture}
	\begin{pgfonlayer}{nodelayer}
		\node [label={[gphase]right:$\pi/2$}, style=gn] (0) at (-2.75, 1.75) {};
		\node [label={[gphase]right:$-\pi/2$}, style=gn] (1) at (-6, 0.5) {};
		\node [label={[rphase]right:$\pi/2$}, style=rn] (2) at (-2.75, 0.5) {};
		\node [label={[rphase]right:$-\pi/2$}, style=rn] (3) at (-6, -0.75) {};
		\node [style=gn] (4) at (-2.75, 3) {};
		\node [label={[gphase]right:$\pi/2$}, style=gn] (5) at (-2.75, -0.75) {};
		\node [label={[gphase]right:$-\pi/2$}, style=gn] (6) at (-2.75, -2) {};
		\node [style=rn] (7) at (-2.75, -3.25) {};
		\node [style=none] (8) at (0, -0) {$=$};
		\node [style=rn] (9) at (2.75, -1.5) {};
		\node [style=gn] (10) at (2.75, 1.5) {};
		\node [style=gn] (11) at (1, 0.5) {};
		\node [style=rn] (12) at (1, -0.5) {};
		\node [style=gn, label={[gphase]right:$-\pi/2$}] (13) at (2.75, -0.5) {};
		\node [style=gn] (14) at (1.75, 0.5) {};
		\node [style=rn] (15) at (1.75, -0.5) {};
		\node [style=Hadamard] (16) at (2.75, 0.5) {};
		\node [style=rn] (17) at (8.5, 1) {};
		\node [label={[gphase]right:$-\pi/2$}, style=gn] (18) at (8.5, -0) {};
		\node [style=gn] (19) at (7.75, 0.5) {};
		\node [style=rn] (20) at (7, -0.5) {};
		\node [style=rn] (21) at (7.75, -0.5) {};
		\node [style=gn] (22) at (7, 0.5) {};
		\node [style=none] (23) at (6, -0) {$=$};
		\node [style=rn] (24) at (8.5, -1) {};
		\node [style=rn] (25) at (13.5, -0.5) {};
		\node [style=rn] (26) at (12.75, -0.5) {};
		\node [style=none] (27) at (11.75, -0) {$=$};
		\node [style=gn] (28) at (12.75, 0.5) {};
		\node [style=gn] (29) at (13.5, 0.5) {};
		\node [style=rn] (30) at (14.25, -0.5) {};
		\node [style=gn] (31) at (14.25, 0.5) {};
		\node [style=gn] (32) at (15.5, -0.25) {};
		\node [style=rn] (33) at (15.5, -1.25) {};
		\node [label={[gphase]right:$-\pi/2$}, style=gn] (34) at (15.5, 1.25) {};
	\end{pgfonlayer}
	\begin{pgfonlayer}{edgelayer}
		\draw (1) to (3);
		\draw (4) to (7);
		\draw (11) to (12);
		\draw (10) to (9);
		\draw (14) to (15);
		\draw (22) to (20);
		\draw (17) to (24);
		\draw (19) to (21);
		\draw (28) to (26);
		\draw (29) to (25);
		\draw (31) to (30);
		\draw [bend right=45, looseness=1.00] (34) to (32);
		\draw (32) to (33);
		\draw [bend left=45, looseness=1.00] (34) to (32);
	\end{pgfonlayer}
\end{tikzpicture}